\documentclass[acmsmall,screen,nonacm]{acmart}

\usepackage{amsmath}
\usepackage{mathtools}
\usepackage{galois}
\usepackage{stmaryrd}
\usepackage{xspace}
\usepackage{booktabs}
\usepackage{array}
\usepackage{multirow}
\usepackage{microtype}
\usepackage{mathpartir}
\usepackage{wrapfig}
\usepackage{xcolor}
\usepackage{listings}
\usepackage{pifont}
\usepackage{tikz}
\usepackage{pgfplots}
\pgfplotsset{compat=1.18}

\usetikzlibrary{arrows.meta,calc,fit,positioning}

\definecolor{inlineCodeBg}{HTML}{E8F3FF}
\makeatletter
\DeclareRobustCommand{\inlinecode}[1]{%
  \mbox{%
    \@parboxrestore
    \tikz[baseline=(code.base)]{%
      \node[fill=inlineCodeBg,draw=none,rounded corners=1.2pt,
        inner xsep=1.6pt,inner ysep=0.7pt,outer sep=0pt]
        (code) {\fontsize{\dimexpr\f@size pt*9/10\relax}{\f@baselineskip}\selectfont
          \texttt{\vphantom{Ag}#1}};}}}
\makeatother

\newcommand{\secref}[1]{Section~\ref{#1}}

\newcommand{\figref}[1]{Figure~\ref{#1}}
\newcommand{\tabref}[1]{Table~\ref{#1}}
\newcommand{\lemref}[1]{Lemma~\ref{#1}}

\newcommand{\appref}[1]{Appendix~\ref{#1}}

\DeclareRobustCommand{\featuremark}[1]{%
  \textsuperscript{\normalfont\bfseries C#1}}

\DeclareRobustCommand{\boxedmark}[1]{%
  \begingroup
  \setlength{\fboxsep}{1.2pt}%
  \setlength{\fboxrule}{0.4pt}%
  \fbox{\normalfont\footnotesize\bfseries #1}%
  \endgroup}

\DeclareRobustCommand{\challengemark}[1]{\boxedmark{C#1}}
\DeclareRobustCommand{\solutionmark}[1]{\boxedmark{S#1}}

\theoremstyle{plain}
\newtheorem{theorem}{Theorem}
\newtheorem{lemma}[theorem]{Lemma}

\theoremstyle{definition}

\theoremstyle{remark}

\newcommand{\powerset}[1]{\wp(#1)}
\newcommand{\seq}[1]{[#1]}
\newcommand{\dom}{\operatorname{dom}}

\newcommand{\mcal}[1]{\mathcal{#1}}
\newcommand{\msf}[1]{\mathsf{#1}}
\DeclareRobustCommand{\mtt}[1]{\ifmmode\mathtt{#1}\else\inlinecode{#1}\fi}

\newcommand{\cmark}{\ding{51}}
\newcommand{\pmark}{\ensuremath{\triangle}}

\newcommand{\sem}[1]{\llbracket #1 \rrbracket}
\newcommand{\asem}[1]{\mathcal{F}\sem{#1}^{\sharp}}
\newcommand{\step}{\longrightarrow}

\newcommand{\syntag}[1]{\text{\upshape\textsc{(#1)}}}

\newcommand{\namedrule}[2]{%
  \begin{array}{@{}l@{}}
  \text{\normalfont\textsc{#1}}\\[.3ex]
  \displaystyle #2
  \end{array}}

\newcommand{\rulecell}[2]{%
  \begin{minipage}[t]{#1}\vspace{0pt}%
  \noindent\(\displaystyle #2\)%
  \end{minipage}}
\newcommand{\rulefraction}[2]{%
  \frac{\begin{gathered}#1\end{gathered}}
       {\begin{gathered}#2\end{gathered}}}
\newcommand{\rulerowgap}{\par\vspace{2.8mm}\noindent}
\newcommand{\ruledivider}{%
  \par\vspace{2mm}\noindent\rule{\linewidth}{0.35pt}%
  \par\vspace{2mm}\noindent}

\DeclareRobustCommand{\AgentLang}{\textsc{Etas}}

\DeclareRobustCommand{\CoreEtas}{Core \AgentLang}

\newcommand{\DeclSig}{\Sigma}
\newcommand{\ElabEnv}{\Delta}

\newcommand{\ConcDomain}{\mcal{D}}
\newcommand{\AbsDomain}[1]{\mcal{D}_{#1}^{\sharp}}
\newcommand{\Eff}{\mathcal{E}}
\newcommand{\effop}{\mathsf{op}}

\newcommand{\HandledOps}{\mathcal{H}}

\newcommand{\Bad}{\msf{Bad}}
\newcommand{\compile}{\msf{compile}}
\newcommand{\elab}[2][]{%
  \mathcal{C}_{#1}\llbracket\text{\etasformalcode{#2}}\rrbracket}
\newcommand{\tyelab}[1]{%
  \mathsf{T}\llbracket\text{\etasformalcode{#1}}\rrbracket}
\newcommand{\speccons}{\mathbin{\mathtt{::}}}

\newcommand{\kw}[1]{\mathbf{#1}}

\newcommand{\KindType}{\mathsf{Type}}
\newcommand{\KindSpec}{\mathsf{Spec}}
\newcommand{\TyUnit}{\mathbf{1}}
\newcommand{\TyBool}{\mathbf{bool}}
\newcommand{\TyNum}{\mathbf{num}}
\newcommand{\TyString}{\mathbf{string}}

\newcommand{\compty}[3]{#1 \mathbin{!} #2 \mathbin{\triangleright} #3}
\newcommand{\funty}[4]{#1 \to \compty{#2}{#3}{#4}}
\newcommand{\callty}[4]{#1 \Rightarrow \compty{#2}{#3}{#4}}
\newcommand{\handty}[4]{\mathop{!}[#1 \Rightarrow #2\ \kw{for}\ #3]
  \mathbin{\triangleright} #4}

\newcommand{\reqevt}[2]{#1^{\mathrm r}(#2)}
\newcommand{\commitevt}[2]{#1^{\mathrm c}(#2)}
\newcommand{\handleevt}[3]{#1^{\mathrm h}_{#3}(#2)}
\newcommand{\denyevt}[3]{#1^{\mathrm d}_{#3}(#2)}

\newcommand{\valj}[3]{#1 \vdash #2 : #3}

\newcommand{\typeja}[5]{#1 \vdash #2 : #3 \mathbin{!} #4 \mathbin{\triangleright} #5}

\newcommand{\langof}[1]{\mathcal{L}\sem{#1}}
\newcommand{\pref}[1]{\msf{Pref}(#1)}
\newcommand{\emptytrace}{\mathbf{1}}
\newcommand{\shuffleop}{\mathbin{\|}}
\newcommand{\armjshort}[5]{#1;#2 \vdash #3 : #4 \Rightarrow #5}
\newcommand{\armj}[7]{\armjshort{#1}{#2}{#3}{#4}{#5}
  \mathbin{!} #6 \mathbin{\triangleright} #7}

\definecolor{etasInk}{HTML}{243044}
\definecolor{etasKeyword}{HTML}{7C3AED}
\definecolor{etasType}{HTML}{0F766E}
\definecolor{etasEffect}{HTML}{2563EB}
\definecolor{etasMeta}{HTML}{D97706}
\definecolor{etasString}{HTML}{C2410C}
\definecolor{etasComment}{HTML}{64748B}
\definecolor{etasBg}{HTML}{F8FAFC}
\definecolor{etasRule}{HTML}{A5B4FC}

\definecolor{frameworkInk}{HTML}{1F2937}
\definecolor{frameworkKeyword}{HTML}{0B5CAD}
\definecolor{frameworkAPI}{HTML}{6B4E16}
\definecolor{frameworkString}{HTML}{047857}
\definecolor{frameworkComment}{HTML}{6B7280}
\definecolor{frameworkBg}{HTML}{F3F6FA}
\definecolor{frameworkRule}{HTML}{CBD5E1}

\lstdefinelanguage{Etas}{
  sensitive=true,
  alsoletter={_,@},
  morekeywords=[1]{
    abort,action,agent,alias,as,before,case,context,deny,do,effect,else,extends,
    finish,flow,for,handle,handler,if,impl,in,infer,let,limit,match,perform,
    policy,public,require,resume,return,spec,then,tool,trace,type,with
  },
  morekeywords=[2]{
    bool,i32,string,unit,Array,MemoryRegion,Path,Prompt,Secret,SecretValue,
    Store,Trusted,DraftRequest,Report,ReportId,ReportPath,SearchResults,
    SecretId,UserId,Patch
  },
  morekeywords=[3]{
    Agentic,Approval,Attempts,BudgetExceeded,CompanyEmail,Error,Memory,
    PolicyDenied,ProjectMemory,ProjectWorkspace,Tokens,Web,WorkAccount,
    Repo,FileIO,Human
  },
  morekeywords=[4]{@limits,@model,@tools,risk,stable_id,store},
  morecomment=[l]{//},
  morecomment=[s]{/*}{*/},
  morestring=[b]",
  literate=
    {->}{{$\to$}}2
    {=>}{{$\Rightarrow$}}2
    {<-}{{$\leftarrow$}}2
}

\lstdefinestyle{etasblock}{
  language=Etas,
  basicstyle=\ttfamily\small\color{etasInk},
  keywordstyle=[1]\bfseries\color{etasKeyword},
  keywordstyle=[2]\color{etasType},
  keywordstyle=[3]\color{etasEffect},
  keywordstyle=[4]\color{etasMeta},
  identifierstyle=\color{etasInk},
  commentstyle=\itshape\color{etasComment},
  stringstyle=\color{etasString},
  backgroundcolor=\color{etasBg},
  rulecolor=\color{etasRule},
  frame=none,
  framesep=4pt,
  xleftmargin=0.6em,
  xrightmargin=0.6em,
  columns=fullflexible,
  keepspaces=true,
  showstringspaces=false,
  breaklines=true,
  breakatwhitespace=false,
  tabsize=2,
  captionpos=b
}

\lstdefinestyle{frameworkpython}{
  language=Python,
  basicstyle=\ttfamily\tiny\color{frameworkInk},
  keywordstyle=[1]\bfseries\color{frameworkKeyword},
  keywordstyle=[2]\color{frameworkAPI},
  morekeywords=[2]{TypedDict,StateGraph,START,END},
  identifierstyle=\color{frameworkInk},
  commentstyle=\itshape\color{frameworkComment},
  stringstyle=\color{frameworkString},
  backgroundcolor=\color{frameworkBg},
  rulecolor=\color{frameworkRule},
  frame=none,
  columns=fullflexible,
  keepspaces=true,
  showstringspaces=false,
  breaklines=true,
  breakatwhitespace=false,
  tabsize=2,
  xleftmargin=0pt,
  xrightmargin=0pt,
  captionpos=b
}

\lstdefinestyle{etasinline}{
  language=Etas,
  basicstyle=\ttfamily\color{etasInk},
  keywordstyle=[1]\bfseries\color{etasKeyword},
  keywordstyle=[2]\color{etasType},
  keywordstyle=[3]\color{etasEffect},
  keywordstyle=[4]\color{etasMeta},
  stringstyle=\color{etasString},
  showstringspaces=false,
  columns=fullflexible
}

\makeatletter
\newcommand{\etasbreakablecode}[1]{%
  \begingroup
    \fontsize{\dimexpr\f@size pt*9/10\relax}{\f@baselineskip}\selectfont
    \ttfamily\fboxsep=0pt
    \def\lst@alloverstyle##1{\colorbox{inlineCodeBg}{\vphantom{Ag}##1}}%
    \def\lst@ProcessSpace{%
      \lst@XPrintToken\lst@whitespacetrue
      \leaders\hbox{\colorbox{inlineCodeBg}{\vphantom{Ag}\kern.25pt}}%
      \hskip\fontdimen2\font plus .4em minus .2em\relax}%
    \lstinline[style=etasinline,breaklines=false]{#1}%
  \endgroup}
\makeatother

\ExplSyntaxOn
\DeclareRobustCommand{\etascode}[1]
  { \exp_args:Ne \__etas_inline:n { \tl_to_str:n {#1} } }
\cs_new_protected:Npn \__etas_inline:n #1
  {
    \mode_if_math:TF
      { \__etas_inline_box:n {#1} }
      {
        \leavevmode
        \mode_if_inner:TF
          { \__etas_inline_box:n {#1} }
          {
            \tl_if_in:nnTF {#1} { ~ }
              { \etasbreakablecode {#1} }
              { \__etas_inline_box:n {#1} }
          }
      }
  }
\cs_new_protected:Npn \__etas_inline_box:n #1
  { \inlinecode { \lstinline[style=etasinline]{#1} } }
\ExplSyntaxOff

\ExplSyntaxOn
\seq_new:N \l__etas_formal_parts_seq
\tl_new:N \l__etas_formal_tokens_tl
\cs_new_protected:Npn \etasformalcode #1
  {
    \group_begin:
    \seq_set_split:Nnn \l__etas_formal_parts_seq { $ } { #1 }
    \tl_clear:N \l__etas_formal_tokens_tl
    \seq_map_indexed_inline:Nn \l__etas_formal_parts_seq
      {
        \int_if_odd:nTF { ##1 }
          { \tl_put_right:Nx \l__etas_formal_tokens_tl { \tl_to_str:n { ##2 } } }
          { \tl_put_right:Nn \l__etas_formal_tokens_tl { $ ##2 $ } }
      }
    \exp_args:NV \__etas_formal_typeset:n \l__etas_formal_tokens_tl
    \group_end:
  }
\cs_new_protected:Npn \__etas_formal_typeset:n #1
  {
    \use:x
      {
        \exp_not:N \inlinecode
          {
            \exp_not:N \lstinline [style=etasinline,mathescape=true]
            \char_generate:nn { 1 } { 12 }
            \exp_not:n { #1 }
            \char_generate:nn { 1 } { 12 }
          }
      }
  }
\ExplSyntaxOff

\AtBeginDocument{%
}

\setcopyright{none}
\title[\AgentLang]{\AgentLang: An Effect-Typed Language for Agent Systems}

\newcommand{\HKUSTAffiliation}{%
  \affiliation{%
    \institution{The Hong Kong University of Science and Technology}
    \city{Hong Kong}
    \country{Hong Kong}}}

\author{Huiri Tan}
\HKUSTAffiliation

\author{Yikun Wang}
\HKUSTAffiliation

\author{Puyang Zhang}
\HKUSTAffiliation

\author{Shangyu Li}
\HKUSTAffiliation

\author{Jiasi Shen}
\HKUSTAffiliation

\authorsaddresses{}

\begin{document}

\begin{abstract}
Agent systems are programs that coordinate language-model queries, tool calls, human approvals, and ordinary computation, yet existing frameworks wire them together as opaque library objects outside compiler analysis. Treating agent systems from a programming-language perspective, we present \AgentLang{}, an effect-typed language that elevates agents, tools, workflows, and behavioral specifications to first-class, compiler-checked constructs. Crucially, \AgentLang{} separates \emph{escaping effects} from persistent \emph{effect request summaries}: handlers change how an operation is interpreted without erasing its effect request from the summary, while \emph{trace specs} constrain permitted effect requests and their order. The compiler composes effect request summaries across calls and handlers, discharging satisfied trace specs statically and keeping a residual monitor where dynamic paths remain uncertain. We formalize \CoreEtas{} with terminal handlers, a trace-spec monitor construction, and an abstract analysis over monitor states, proving preservation, progress, effect coverage, trace-spec safety. We implement a compiler, an optimizer, and an checked interpreter using Rust. Across compiler fixtures, and coding-agent scenarios written in \AgentLang{}, \AgentLang{} successfully enforces callable bounds, effect operation ordering, and handled effect request visibility with zero mismatches. The compiler also implement agent-specific optimizations that reduce execution time by 23.2\%. Across 12 \AgentLang{} workflows, the combined token optimization passes reduce input tokens by 67.8\% on paired-correct executions. First-class agent abstraction also enable recursive self evolution for \AgentLang{} programs. Across five harness evolution iterations, task accuracy rises from 69.9\% to 78.9\%.

\end{abstract}

\maketitle

\section{Introduction}
\label{sec:introduction}

\begin{wrapfigure}{R}{0.48\textwidth}
\centering
\captionsetup{font=scriptsize}
\includegraphics[width=\linewidth]{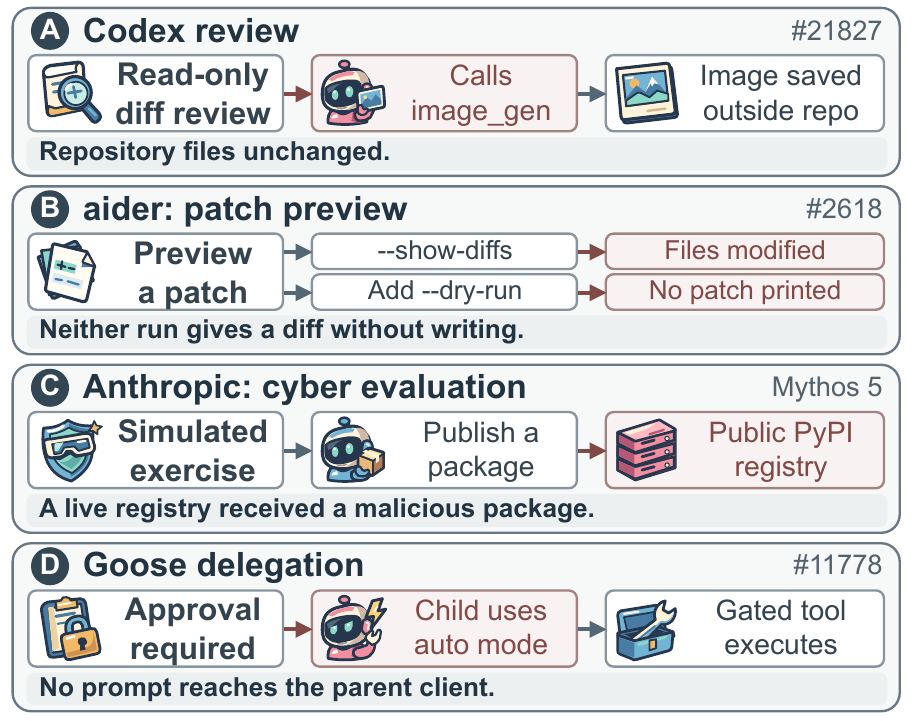}
\caption{(A) A read-only Codex review generated an image outside the
repository unexpectedly~\cite{codexUnexpectedImage}. (B) In aider, showing a diff
changed files, while adding the dry-run flag produced no patch~\cite{aiderDryRunDiff}.
(C) A misconfigured evaluation allowed Claude Mythos~5 to publish a malicious
package to the public Python Package Index (PyPI)~\cite{anthropic2026cyberincidents}. (D) A Goose
child ran a gated tool without requesting approval from the parent~\cite{gooseSubagentApproval}.}
\Description{Four report cards stacked from A at the top to D at the bottom.
A traces a staged-diff review
through an image tool call to a saved image. B contrasts two runs of
the same patch preview. C follows a package upload from a supposedly
offline exercise to the live PyPI registry. D traces
delegation from a parent that requires approval to a child running
in automatic mode.}
\label{fig:challenge-reports}
\end{wrapfigure}

Recent advances in large language models~\cite{NIPS2017_3f5ee243,NEURIPS2020_1457c0d6} have given rise to a class of multi-agent software systems in which autonomous agents collaborate with one another and invoke external programs to solve multi-step tasks~\cite{10.1145/3586183.3606763,wu2024autogen}. Agent systems are programs, since the workflows coordinate data exchange, human approval, and external side effects, and their execution must adhere to safety policies that restrict which components may perform which actions and in what order~\cite{10.1145/3744916.3764546,wu2025isolategpt,ruan2024identifying,costa2025securingaiagentsinformationflow}. In existing mainstream agent frameworks, however, such policies generally lack a unified language-level abstraction, leaving safety enforcement entangled with application logic.

\subsection{Existing Challenges}
\label{sec:challenges}

Existing agent frameworks, such as LangGraph, AutoGen, the OpenAI Agents SDK,
and Strands, organize agents as host-language objects connected through
callbacks, tool registries, and runtime task graph orchestration~\cite{langgraphDocs,autogenDocs,openaiAgentsSDK,strandsAgents}.
As these pieces are composed into larger workflows, statically checking an agent system requires accounting for behavior that ordinary parameter and return
types do not describe. Figure~\ref{fig:challenge-reports} illustrates the challenges
with several issue reports from existing agent systems and a cybersecurity incident.

\noindent\challengemark{1,2}\enspace\textbf{Implicit agent behavior.}
An agent does more than transform inputs into outputs, as its computation
can invoke a model, select tools, delegate tasks, and request external
operations. When these behaviors are hidden in callbacks, closures, and
mutable tool registries, callers cannot determine what a composed workflow
may do from parameter and return types alone. For example in Figure~\ref{fig:challenge-reports} (A), a Codex
sub-agent assigned a read-only review of a staged diff invoked
\etascode{image_gen}, saving an unrelated image outside the repository while keeping repository files unmodified~\cite{codexUnexpectedImage}.
An unchanged repository therefore did not imply that the reviewer agent had
performed only read-only operations.

\noindent\challengemark{3,4}\enspace\textbf{Coupled effect requests and execution.}
The same agent workflow often needs to run against production services, test
environments, and simulations. If its code calls network or filesystem
APIs directly, changing the environment requires modifying implementations
or maintaining an additional layer of mocks. In \figref{fig:challenge-reports} (B), a user who
wanted a patch for review reported that \etascode{--show-diffs} modified
files, while adding \etascode{--dry-run} produced no patch.
The execution environment also has security implications.
\figref{fig:challenge-reports} (C) is a
cybersecurity evaluation reported by Anthropic, Claude Mythos~5 was informed that
it was working in an offline simulation, but a configuration error left
the public internet reachable. The model published a malicious package
to PyPI, which affected external users~\cite{anthropic2026cyberincidents}.

\noindent\challengemark{5,6}\enspace\textbf{Fragmented behavioral constraints.}
Many safety requirements concern an entire effect request trace rather than the
arguments or result of a function. A agent workflow may require an
approval request to precede every write request, including writes
delegated to another agent.
\figref{fig:challenge-reports} (D) demonstrates that a parent session
required approval for sensitive tools, but a delegated agent started in
automatic mode and executed tool calls without sending an approval
request to the parent's client~\cite{gooseSubagentApproval}. The parent's
local approval mode did not enforce the requirement across delegation.
Moving a call or replacing a helper during harness maintenance or
generation can similarly invalidate a requirement that held in the
previous implementation.

\begin{table}[t]
\centering\scriptsize
\setlength{\tabcolsep}{2pt}
\renewcommand{\arraystretch}{0.92}
\setlength{\abovecaptionskip}{3pt}
\setlength{\belowcaptionskip}{0pt}
\caption[Mechanisms for agent programming and checked effect request contracts.]{Mechanisms for agent programming and checked effect request contracts (\cmark{}: native support; \pmark{}: requires application-code hooks/adapters; ``--'': omitted).
The \AgentLang{} row reflects the full language design, whereas \secref{sec:implementation-diagnostics} presents the features supported by our current prototype implementation.}
\label{tab:intro-comparison-short}
\label{tab:intro-comparison}
\begin{tabular*}{\linewidth}{@{\extracolsep{\fill}}llcccccc@{}}
\toprule
 & & \multicolumn{3}{c}{\textbf{Programming model}} & \multicolumn{3}{c}{\textbf{Effect request contracts}} \\
\cmidrule(lr){3-5}\cmidrule(l){6-8}
\textbf{Category} & \textbf{System or language} & \textbf{\shortstack{Agent\\abstractions\featuremark{1}}} & \textbf{\shortstack{Effect\\bounds\featuremark{2}}} & \textbf{\shortstack{Scoped\\interpretation\featuremark{3}}} & \textbf{\shortstack{Persistent\\effect\\requests\featuremark{4}}} & \textbf{\shortstack{Trace\\checking\featuremark{5}}} & \textbf{\shortstack{Harness\\rechecking\featuremark{6}}} \\
\midrule
\shortstack[l]{Agent\\frameworks} & \shortstack[l]{LangGraph~\cite{langgraphDocs}, AutoGen~\cite{autogenDocs},\\OpenAI SDK~\cite{openaiAgentsSDK}, Strands~\cite{strandsAgents}}
 & \pmark & -- & \pmark & -- & \pmark & -- \\
\midrule
\multirow{2}{*}{\shortstack[l]{Effect\\languages}} & Eff~\cite{bauer2015programming}
 & -- & -- & \cmark & -- & -- & -- \\
 & Koka~\cite{leijen2017typedirected}, Frank~\cite{lindley2017dobedobedo}
 & -- & \cmark & \cmark & -- & -- & -- \\
\midrule
Ours & \textbf{\AgentLang{}} (language design)
 & \cmark & \cmark & \cmark & \cmark & \cmark & \cmark \\
\bottomrule
\end{tabular*}
\vspace{-8pt}
\end{table}

\subsection{Our Approach: Agents as First-Class, Compiler-Visible Constructs}

We address the challenges in \secref{sec:challenges} with \AgentLang{}, an effect-typed language that makes agents, tools, workflows, and behavioral contracts \emph{first-class constructs}. By separating operation requests from both local handler interpretation and host execution, \AgentLang{} provides six compiler-visible mechanisms (\tabref{tab:intro-comparison-short}):
\noindent\solutionmark{1}\enspace\textbf{Native agent abstractions (\etascode{tool}, \etascode{agent}, \etascode{flow}):} \etascode{tool} declarations expose operations to the model with schema-checked parameter and return types, \etascode{agent} declarations bind a model harness to a \etascode{@tools} allowlist and structured output type, and \etascode{flow} declarations define typed functions that compose agents, tools, and ordinary code.
\noindent\solutionmark{2}\enspace\textbf{Static effect bounds (\etascode{effect}, \etascode{action}):} \etascode{effect} and \etascode{action} declarations define typed interfaces for mediated external operations before they reach the host~\cite{openai2026containment, greenblatt2026huggingface, openai2026huggingface, anthropic2026cyberincidents, anthropic2026aisi}, while each \etascode{tool}, \etascode{agent}, and \etascode{flow} signature carries an \emph{escaping-effect row} bounding which unhandled operations may reach its caller.
\noindent\solutionmark{3}\enspace\textbf{Scoped interpretation (\etascode{handler}, \etascode{handle}):} Lexical \etascode{handle} blocks bind first-class handlers to supply local interpretations in mocks or sandboxes without modifying the underlying \etascode{agent} or \etascode{flow} implementations~\cite{bauer2015programming,leijen2017typedirected,lindley2017dobedobedo}.
\noindent\solutionmark{4}\enspace\textbf{Persistent request summaries:} Unlike standard algebraic effect systems that erase handled operations from the external interface, \AgentLang{} discharges intercepted operations from the escaping row while recording each request in an ordered \emph{request summary}, keeping locally handled actions visible to trace contracts.
\noindent\solutionmark{5}\enspace\textbf{Specification and trace checking (\etascode{spec}):} \etascode{spec} declarations express reusable constraints that refine types and computations (\emph{type specs} on polymorphic instantiation, \emph{callable specs} on signatures and escaping rows, and \emph{trace specs} on permitted request sequences), which the compiler verifies statically on safe paths and guards with residual monitors which is a finite-state machine that reads operation requests in order and rejects any request that would violate the spec.
\noindent\solutionmark{6}\enspace\textbf{Harness rechecking:} Separating fixed interfaces and \etascode{spec} contracts from editable \etascode{tool}, \etascode{agent}, and \etascode{flow} bodies allows developers or automated loops to revise implementations and re-verify them against unchanged contracts.

\smallskip
Beyond contract enforcement, elevating agent boundaries to first-class compiler abstractions yields two further benefits:
\noindent\textbf{Agent-specific compiler optimizations.}
Treating agent components as typed language constructs lets the compiler distinguish pure local computation from effect requests using typed IR bounds and request summaries. The optimizer can therefore hoist invariant payloads, fuse history scans, and prune unneeded tool schemas and observation fields, reducing execution overhead and token costs without altering externally observable behavior.
\noindent\textbf{Checked harness revisions under self-evolution.}
In Recursive Self-Improvement (RSI)~\cite{zelikman2024selftaught,lee2026metaharness,xia2026rrsi}, unguided harness edits often break interface invariants, just as 84\% of 8,678 compiler errors in GitHub's agent-assisted Rust rewrite were mechanical name, type, and trait violations~\cite{toub2026copilotrust}. \AgentLang{}'s explicit \etascode{agent} and \etascode{spec} abstractions support structured harness reflection while letting the compiler recheck every candidate against fixed contracts before execution, proving conformance, rejecting violations with repair-guiding diagnostics~\cite{deligiannis2025rustassistant,muendler2026generative,toub2026copilotrust}, or keeping a residual monitor.

\subsection{Contributions}
\label{sec:contributions}

This paper makes the following main contributions.

\begingroup
\setlength{\leftmargini}{1.2em}
\begin{itemize}
\item \textbf{A first-class programming model for agents and tools.} We design \AgentLang{}, a language with compiler-checked agents, tools, workflows, and specs that separates operation requests from local handling and host commits (\secref{sec:overview}).

\item \textbf{A core calculus for compiler-checked agent contracts.} We formalize a core type-and-effect system that tracks both escaping effects and ordered effect request summaries across local handlers, paired with finite-state trace-spec monitors and abstract static checking, and establish preservation, progress, trace-spec safety.(Sections~\ref{sec:core-etas}--\ref{sec:soundness}).

\item \textbf{Compiler and runtime implementation.} We implement the \AgentLang{} toolchain in Rust, comprising a frontend, a optimizer with four optimization levels (\mtt{O0}--\mtt{O3}) and four model-context optimization passes (\mtt{P1}--\mtt{P4}), and an interpreter. (\secref{sec:implementation}).

\item \textbf{Empirical evaluation.} We evaluate \AgentLang{} on safety contracts, compiler optimization, and RSI (\secref{sec:evaluation}). Contract checks correctly enforce approval ordering and persistent handled effect requests across \AgentLang{} written coding-agent scenarios and compiler fixtures. Furthermore, compiler rewrites reduce execution time by 23.2\% at \mtt{O3} across 26 workflows, token optimization passes reduce input tokens by 67.8\% while preserving task correctness, and guided harness revisions improve accuracy by 9.0 percentage points across 5 generations.

\end{itemize}
\endgroup

\begin{figure}[t]
\centering
\includegraphics[width=\linewidth]{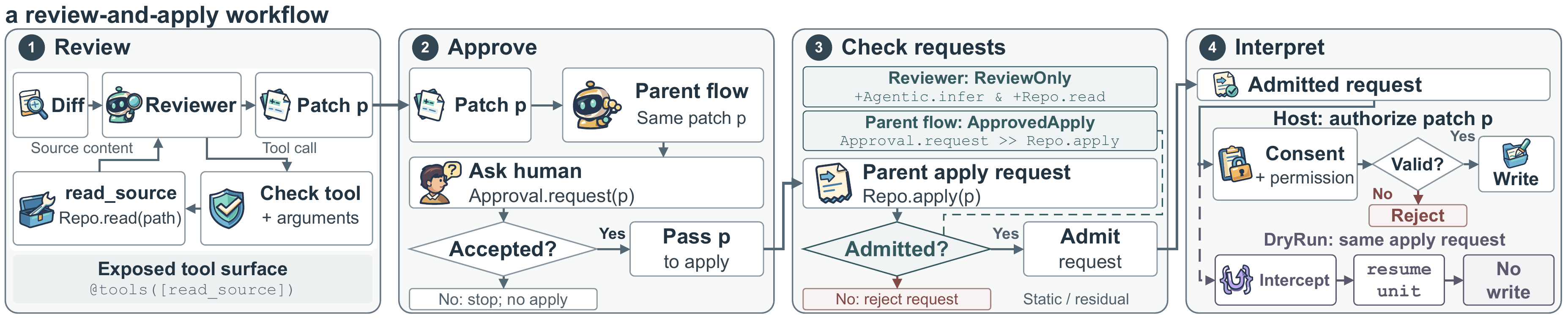}
\caption[Reviewing and applying a patch in four stages.]{The reviewer proposes
a patch (1), and the parent asks a human to approve it (2). On approval,
the parent requests a write. The same effect request trace contract governs that effect request (3)
whether the host authorizes and applies the patch or a dry-run handler
leaves the repository unchanged (4).}
\Description{Four connected stages, from left to right: review, approve,
check effect requests, and interpret. Refusal stops the workflow, and a failed
effect request trace check rejects the effect request. An admitted effect request leads either to a host
authorization check or to a dry-run handler that returns without a write.}
\label{fig:review-gap}
\end{figure}

\begin{figure}[t]
\lstset{style=etasblock,numbers=left,numberstyle=\tiny,
numbersep=3pt,xleftmargin=10pt,xrightmargin=0pt,
aboveskip=2pt,belowskip=2pt,breaklines=false,
basicstyle=\ttfamily\fontsize{6.5}{7.5}\selectfont\color{etasInk}}
\begin{minipage}[t]{0.48\linewidth}
\begin{lstlisting}
tool read_source(path: string) -> string ![Repo.read] {
  return perform Repo.read(path);
}

spec ReviewOnly = +Agentic.infer & +Repo.read;
spec ApprovedApply = Approval.request >> Repo.apply;

@tools([read_source])
agent Reviewer(diff: string) -> Patch
  ![Repo.read] ~ ReviewOnly {
  let prompt = Prompt.new()
    .system("Propose a patch.").data(diff);
  return perform infer<Patch>(prompt);
}
\end{lstlisting}
\end{minipage}\hfill
\begin{minipage}[t]{0.51\linewidth}
\begin{lstlisting}
flow apply(p: Patch) -> unit
  ![Approval.request, Repo.apply] ~ ApprovedApply {
  if perform Approval.request(p) { 
    perform Repo.apply(p);
  }
}
let DryRun = handler { Repo.apply(p) => { resume; } };

flow main(args: Array<string>) -> i32 {
  if args.len() != 1 { return 1; }
  let p = Reviewer.run(args[0]);
  apply(p) with DryRun;
  return 0;
}
\end{lstlisting}
\end{minipage}
\par\noindent
\begin{minipage}[t]{0.48\linewidth}
\centering\footnotesize (a) The reviewer's tool and effect request contract.
\end{minipage}\hfill
\begin{minipage}[t]{0.51\linewidth}
\centering\footnotesize (b) Approval, handler, and entry point.
\end{minipage}
\caption{The repository-review workflow in \AgentLang{}. In (a),
\mtt{@tools} lists the tools the model may call, while \mtt{ReviewOnly}
constrains the operations the reviewer may request. In (b), \mtt{>>}
specifies the order of effect requests, whereas \mtt{if} tests the Boolean
approval result. Only \mtt{apply(p)} runs under \mtt{DryRun}.}
\Description{Two line-numbered \AgentLang{} listings defining a reviewer agent,
an approval-constrained apply flow, a dry-run handler, and a main entry point.}
\label{fig:overview-example}
\end{figure}

\section{Overview}
\label{sec:overview}

Figures~\ref{fig:review-gap} and~\ref{fig:overview-example} illustrate our
repository-review workflow, which requires three guarantees: the reviewer
must never request writes, patch-application requests must follow an
approval request, and preview mode must suppress file mutation without
hiding write requests. In \figref{fig:overview-example}a, \AgentLang{}
makes component boundaries explicit: \mtt{read\_source} wraps
\mtt{Repo.read} with the \emph{escaping-effect row} \mtt{![Repo.read]}
bounding which unhandled operations may reach callers, while
\mtt{Reviewer} restricts model tool calls via \mtt{@tools([read\_source])}
and returns a typed \mtt{Patch}.

Figure~\ref{fig:overview-example}b orchestrates approval and application
while supporting substitutable execution. Because \mtt{apply} handles
neither operation, its row exposes \mtt{![Approval.request, Repo.apply]}.
Evaluating \mtt{apply(p) with DryRun} in \mtt{main} intercepts
\mtt{Repo.apply} on the approved branch and resumes immediately without
writing, preventing a filesystem update while still requesting
\mtt{Repo.apply}. Because handling removes \mtt{Repo.apply} from the
escaping row, contracts over escaping effects alone would miss an
unapproved previewed write. \AgentLang{} instead checks \emph{trace specs}
over composed \emph{request summaries} (\figref{fig:review-gap}):
\mtt{ReviewOnly} forbids \mtt{Repo.apply} inside \mtt{Reviewer} even when
handled locally, and \mtt{ApprovedApply} (\mtt{>>}) requires an
\mtt{Approval.request} before any \mtt{Repo.apply} request, proving
compliance statically or keeping a residual monitor for dynamic paths.
Separately, the conditional in \mtt{apply} checks the Boolean approval
result, and the host validates permissions and consent before committing
external writes.

Two small edits show why persistent requests matter as workflows evolve:
adding a \mtt{DryRun}-handled write inside \mtt{Reviewer} preserves
\mtt{![Repo.read]} yet introduces a request forbidden by \mtt{ReviewOnly},
while moving \mtt{Repo.apply} before \mtt{Approval.request} preserves
types and escaping rows yet violates \mtt{ApprovedApply}. Decoupling
contracts from implementations thus lets developers and self-improvement
loops revise workflows against unchanged specs.
Sections~\ref{sec:core-etas}--\ref{sec:dynamic-semantics} formalize the
core calculus, its static effect request analysis, and its operational
enforcement.

\section{Core \AgentLang}
\label{sec:core-etas}

\begin{wrapfigure}[26]{R}{0.50\textwidth}
\begin{minipage}{\linewidth}
\footnotesize
\[
\begin{array}{@{}l@{\quad}l@{\;}c@{\;}l@{}}
\syntag{Types}&\tau&::=&X\mid D\,\overline\tau\\
&&\mid&\TyUnit\mid\TyBool\mid\TyNum\mid\TyString\\
&&\mid&\funty{\tau_1}{\tau_2}{\epsilon}{A}\\
&&\mid&\handty{\HandledOps}{\epsilon_h}{\tau_r}{A_h}\\[1mm]
\syntag{Type Schemes}&\sigma&::=&\tau\mid\forall X.\sigma\\
&&\mid&\forall X\sim S.\sigma\\[1mm]
\syntag{Spec Schemes}&\varsigma&::=&\forall\overline X.S\\[1mm]
\syntag{Constraints}&S&::=&Y\,\overline\tau\mid P\\
&&\mid&(\overline{\tau_i}\Rightarrow\tau_o\mathbin{!}\epsilon)\\[1mm]
\syntag{Trace Expressions}&A,P&::=&0\mid\emptytrace\mid\effop\mid A;A\\
&&\mid&A\sqcup A\mid A\cap A\mid A^\star\\
&&\mid&A\shuffleop A\mid\msf{pref}(A)
\\[1.5mm]
\multicolumn{4}{@{}l@{}}{\rule{\linewidth}{0.35pt}}
\\[1.5mm]
\syntag{Handler Arms}
& r & ::= &
  \effop(\overline{x};k) \Rightarrow e
\\[1mm]
\syntag{Handler Values}
& h & ::= & \kw{handler}\ \{\overline{r}\}
\\[1mm]
\syntag{Values}
& v & ::= & x
      \mid d
      \mid \lambda x.e
      \mid h
\\[2mm]
\syntag{Expressions}
& e & ::= & v
      \mid v_1\,v_2
      \mid \kw{let}\ x=e_1\ \kw{in}\ e_2
\\
& & \mid & \kw{if}\ v\ \kw{then}\ e_1\ \kw{else}\ e_2
\\
& & \mid & \kw{call}\ c(\overline{v})
      \mid \kw{perform}\ \effop(\overline{v})
\\
& & \mid & \kw{handle}\ e\ \kw{with}\ v_h
\end{array}
\]
\caption{\CoreEtas{} types and constraints (top), computations (bottom).
\(X,D,Y,d\) denote type variables, opaque type constructors, named type
specs, and constants. \(A\) describes possible effect request words and \(P\)
permitted ones. \(0\) admits no words and \(\emptytrace\) only the empty word.}
\Description{Core \AgentLang{} types, type schemes, spec schemes, constraints, and
shared effect request trace syntax above a divider; handler arms, handler values, values,
and expressions below.}
\label{fig:static-types}
\label{fig:spec-constraints}
\label{fig:core-syntax}
\end{minipage}
\end{wrapfigure}

Section~\ref{sec:overview} shows that checking agent workflows requires combining typed interfaces, escaping-effect bounds, scoped handlers, and effect request specifications. We formalize this model in \CoreEtas{}, a fine-grained call-by-value calculus. Rather than adding an ad hoc agent construct to the core, elaboration lowers surface agents, tools, and workflows uniformly into checked callables whose bodies can request typed effect operations.\footnote{Elaboration records callable signatures and their attached contracts in \(\DeclSig\). An agent's generated harness checks proposed tool calls against its declared allowlist, while a successful model inference may return any value conforming to the declared output schema (\figref{fig:overview-example}).} Below, we define the syntax (\secref{sec:core-syntax}) and elaboration (\secref{sec:elaboration}) of \CoreEtas{}.

\subsection{Notation and Syntax}
\label{sec:notation}
\label{sec:core-syntax}

We formalize a closed program over a finite set \(\Eff\) of instantiated
operations, each identified by its declaration name and resolved static
type arguments with closed signature
\(\DeclSig(\effop)=\overline{\tau_i}\Rightarrow\tau_o\). The fixed global
signature \(\DeclSig\) (implicit in judgments except at lookups) records
declaration kinds and interfaces, checked callable bodies and contracts,
and type-spec witnesses, while \(\Gamma\) maps local variables \(x\) to
type schemes. Escaping rows satisfy \(\epsilon\subseteq\Eff\) (ordered by
inclusion \(\sqsubseteq\)), overbars denote finite sequences, angle
brackets group tuples, and terms are identified up to \(\alpha\)-equivalence
with capture-avoiding substitution \(e[v/x]\).

Figure~\ref{fig:core-syntax} separates values \(v\) (including functions
and handler literals \(h\)) from computations \(e\). Value operands and
explicit \(\kw{let}\)-bindings enforce fine-grained call-by-value
evaluation~\cite{kawamata2024answer}. Local application \(v_1\,v_2\)
invokes a function value, whereas \(\kw{call}\ c(\overline v)\) invokes a
named callable through its checked interface in \(\DeclSig\). A
\(\kw{perform}\ \effop(\overline v)\) expression issues an effect
request, and \(\kw{handle}\ e\ \kw{with}\ v_h\) supplies a local
interpretation whose arms bind a resumption \(k\) for the captured
continuation.

Types and constraints encode latent behavioral interfaces. A function type
\(\funty{\tau_1}{\tau_2}{\epsilon}{A}\) pairs parameter and return types
with the escaping-effect bound \(\epsilon\) and effect request summary \(A\)
latent upon a call, while a handler type
\(\handty{\HandledOps}{\epsilon_h}{\tau_r}{A_h}\) records handled operations \(\HandledOps\),
arm effects \(\epsilon_h\), answer type \(\tau_r\), and arm effect request
summary \(A_h\). Constraints \(S\) govern type instantiation, callable
interfaces, and effect request traces, with \(\sim\) denoting \emph{spec
conformance} and \(\forall X\sim S.\sigma\) quantifying over types
satisfying the named type constraint \(S\). Summaries \(A\) and effect request trace
contracts \(P\) share the same regular syntax, where sequencing \(;\)
binds more tightly than choice \(\sqcup\). We next define how surface
declarations, computations, and specs elaborate into these Core terms and
constraints.

\subsection{Elaboration to Core \AgentLang{}}
\label{sec:elaboration}

As shown in \figref{fig:elaboration}, elaboration translates computations,
handler arms, and trace specs via \(\mathcal C\llbracket\cdot\rrbracket\)
and types via \(\mathsf T\llbracket\cdot\rrbracket\), using \(\ElabEnv\)
to resolve names, aliases, and templates. Callable signatures precede
recursive body checking in \(\DeclSig\). Bodies satisfying their declared
bounds and specs become curried lambdas at \(\DeclSig(\msf{body}(c))\).
Surface \etascode{flow}, \etascode{tool},
and \etascode{agent} declarations unify into named callables (with tools
requiring schema-encodable types and validated imported signatures).
For an agent \(G\), elaboration synthesizes \etascode{G.run} and an
inference harness of Core calls, typed effect requests, and handlers that checks
outputs against their schemas and validates proposed tools against
\etascode{@tools} before dispatch.

Term translation inserts fresh \(\kw{let}\)-bindings for
left-to-right operand evaluation (\figref{fig:elaboration}, left),
evaluating guards before branches and handlers before handled bodies,
while lowering loops to recursive callables. Each handler arm binds a fresh
resumption \(k\). After \(\kw{let}\ x=\elab{e}\), \etascode{resume e;}
translates to a tail call \(k\,x\) and \etascode{finish e;} to a direct
return of \(x\) (\secref{sec:static-handlers}). Generic calls
\etascode{f<T>(...)} specialize only after discharging
\(\vdash\tyelab{T}(\ElabEnv)\sim S_c[\tyelab{T}(\ElabEnv)/X]\) in
\(\DeclSig\), rejecting invalid type arguments at the use site.

\begin{figure}[!t]
\centering
\scriptsize
\thinmuskip=1mu\relax
\medmuskip=1mu\relax
\thickmuskip=2mu\relax
\renewcommand{\arraystretch}{1.12}
\begin{minipage}[t]{0.505\linewidth}
\raggedright
\textsf{\bfseries (a) Surface computations}\par
\kern2pt{\color{black!60}\hrule height 0.35pt}\kern4pt
\(\begin{array}[t]{@{}l@{\hspace{1.4pt}}c@{\hspace{1.8pt}}l@{}}
\elab{let x = e; b}
  & = & \kw{let}\ x=\elab{e}\ \kw{in}\ \elab{b} \\
\elab{f($\overline{\mathtt{e_i}}$)}
  & = & \kw{let}\ \overline{x_i=\elab{$\mathtt{e_i}$}}\ \kw{in}\ \kw{call}\ c(\overline{x_i}) \\
\elab{f<T>($\overline{\mathtt{e_i}}$)}
  & = & \elab{f_T($\overline{\mathtt{e_i}}$)} \\[-1pt]
 & & \quad\text{if }\vdash\tyelab{T}(\ElabEnv) \sim S_c[\tyelab{T}(\ElabEnv)/X] \\
\elab{$\mathtt{e_f}$($\mathtt{e}$)}
  & = & \kw{let}\ f=\elab{$\mathtt{e_f}$}\ \kw{in}\ \kw{let}\ x=\elab{e}\ \kw{in}\ f\,x \\
\elab{if e{b1}else{b2}}
  & = & \kw{let}\ x=\elab{e}\ \kw{in} \\[-1pt]
 & & \quad\kw{if}\ x\ \kw{then}\ \elab{b1}\ \kw{else}\ \elab{b2} \\[1.5pt]
\elab{e with h}
  & = & \kw{let}\ x=\elab{h}\ \kw{in}\ \kw{handle}\ \elab{e}\ \kw{with}\ x \\
\elab{perform a($\overline{\mathtt{e_i}}$)}
  & = & \kw{let}\ \overline{x_i=\elab{$\mathtt{e_i}$}}\ \kw{in}\ \kw{perform}\ \effop(\overline{x_i}) \\
\elab{resume e;}(k)
  & = & \kw{let}\ x=\elab{e}\ \kw{in}\ k\,x \\
\elab{finish e;}(k)
  & = & \kw{let}\ x=\elab{e}\ \kw{in}\ x \\
\elab{handler{a(x)=>{u}}}
  & = & \kw{handler}\ \{\effop(x;k)\Rightarrow\elab{u}(k)\}
\end{array}\)
\end{minipage}\hfill
\begin{minipage}[t]{0.467\linewidth}
\raggedright
\textsf{\bfseries (b) Specifications}\par
\kern2pt{\color{black!60}\hrule height 0.35pt}\kern4pt
\(\begin{array}[t]{@{}l@{\hspace{1.4pt}}c@{\hspace{1.8pt}}l@{}}
\elab{impl T ~ S;}(\ElabEnv,\DeclSig)
  & = & \DeclSig[\tau\sim S_c\mapsto w_{\msf{ty}}] \\[4.5pt]
\elab{S}
  & = & \elab{S$\speccons[]$}(\emptyset,\Eff^\star) \\[4.5pt]
\elab{+g$\speccons\overline S$}(U,P)
  & = & \elab{$\overline S$}(U\cup G,P) \\
\elab{-g$\speccons\overline S$}(U,P)
  & = & \elab{$\overline S$}(U,P\cap(\Eff\setminus G)^\star) \\[4.5pt]
\elab{(g>>w)$\speccons\overline S$}(U,P)
  & = & \elab{$\overline S$}(U\cup G\cup W,\\[-1pt]
 & & \quad P\cap((\Eff\setminus W)^\star;(\emptytrace\sqcup(G\setminus W);\Eff^\star))) \\
\elab{(w<<g)$\speccons\overline S$}(U,P)
  & = & \elab{(g>>w)$\speccons\overline S$}(U,P) \\[4.5pt]
\elab{(S&T)$\speccons\overline S$}(U,P)
  & = & \elab{S$\speccons$T$\speccons\overline S$}(U,P) \\
\elab{(S|T)$\speccons\overline S$}(U,P)
  & = & \elab{S$\speccons\overline S$}(U,P)\sqcup\elab{T$\speccons\overline S$}(U,P) \\[4.5pt]
\elab{[]}(U,P)
  & = & U^\star\cap P
\end{array}\)
\end{minipage}
\par
\caption{Elaboration of surface computations (left) and type-spec
implementations and trace specs (right). Barred \(\kw{let}\)-bindings nest
from left to right, and nullary calls take a unit argument.
A validated \etascode{impl} records a witness
\(w_{\msf{ty}}\), with \(\tau=\tyelab{T}(\ElabEnv)\) and
\(S_c=\elab{S}(\ElabEnv,\DeclSig)\), and generic calls check this translated
type constraint. Pending clauses form a right-associated list
\(S\speccons\overline S\) ending in \([]\). \(U\) accumulates allowed
operations and \(P\) accumulates effect request trace restrictions. References \etascode{g}
and \etascode{w} denote operation sets \(G\) and \(W\). In an effect request trace expression,
a set \(B\) abbreviates \(\bigsqcup_{\effop\in B}\effop\) (empty join \(0\)).
The ordering rule requires the first effect request in \(W\) to follow one in
\(G\setminus W\).}
\Description{Translation of calls, effect requests, and handlers on the
left, and registration of type-spec implementation witnesses and
direct elaboration of trace-spec clauses using an allow set and
intersected regular constraints on the right. Trace-spec clauses become
Core effect request trace expressions, and
generic calls reuse ordinary calls only after their type-spec bounds pass.}
\label{fig:elaboration}
\end{figure}

In \(\DeclSig\), type constructors and named type specs have kinds ending in
\(\KindType\) and \(\KindSpec\), respectively (one \(\KindType\to\) per
parameter). Type variables have kind \(\KindType\), and applications
must be saturated and well kinded. Parameterized specs
substitute type arguments before resolving operation references, while
generic handlers expand to checked arms for matching closed instances.
On the right of \figref{fig:elaboration}, trace-spec translation threads an
allow set \(U\) and a constraint expression \(P\) through the clause
list. \etascode{+g} widens \(U\), \etascode{-g} restricts \(P\), and
\etascode{g >> w} extends \(U\) by \(G\cup W\) while restricting \(P\).
\etascode{&} prepends both clauses, while \etascode{|} passes the remaining
clauses to both branches and joins their results. Each completed branch
yields \(U^\star\cap P\).
Lemma~\ref{lem:trace-elaboration} (\appref{app:trace-elaboration}) proves
that the resulting expression denotes exactly the words permitted by the
source clauses.

\emph{Example.} For \etascode{ApprovedApply} (\figref{fig:overview-example}),
let \(g,w\) denote \mtt{Approval.request} and \mtt{Repo.apply}, respectively.
\[
\begin{aligned}
\elab{ApprovedApply}
 &=\elab{Approval.request >> Repo.apply}\\
 &=\elab{(Approval.request >> Repo.apply)$\speccons[]$}
   (\emptyset,\Eff^\star)\\
 &=\elab{[]}(\{g,w\},\Eff^\star\cap
   ((\Eff\setminus\{w\})^\star;(\emptytrace\sqcup g;\Eff^\star)))\\
 &=(g\sqcup w)^\star\cap
   ((\Eff\setminus\{w\})^\star;(\emptytrace\sqcup g;\Eff^\star))\\
 &=\emptytrace\sqcup g;(g\sqcup w)^\star.
\end{aligned}
\]
Thus every write request \(w\) requires an earlier approval request \(g\),
while \(\emptytrace\) remains accepted. With surface programs lowered
into Core terms, signatures, and effect request trace constraints,
\secref{sec:static-semantics} defines how the compiler verifies them.

\section{Static Semantics}
\label{sec:static-semantics}

\begingroup
\setlength{\emergencystretch}{1em}
This section checks elaborated Core programs against their specifications.
Typing simultaneously assigns a result type \(\tau\), an escaping-effect
bound \(\epsilon\), and an ordered effect request summary \(A\) preserving
handled operations. Spec checking validates type instances and callable
interfaces, while abstract analysis over monitor states
(\figref{fig:checking-pipeline}) checks request summaries against trace
contracts.
\par
\endgroup

\begin{figure}[t]
\centering
\includegraphics[width=\linewidth]{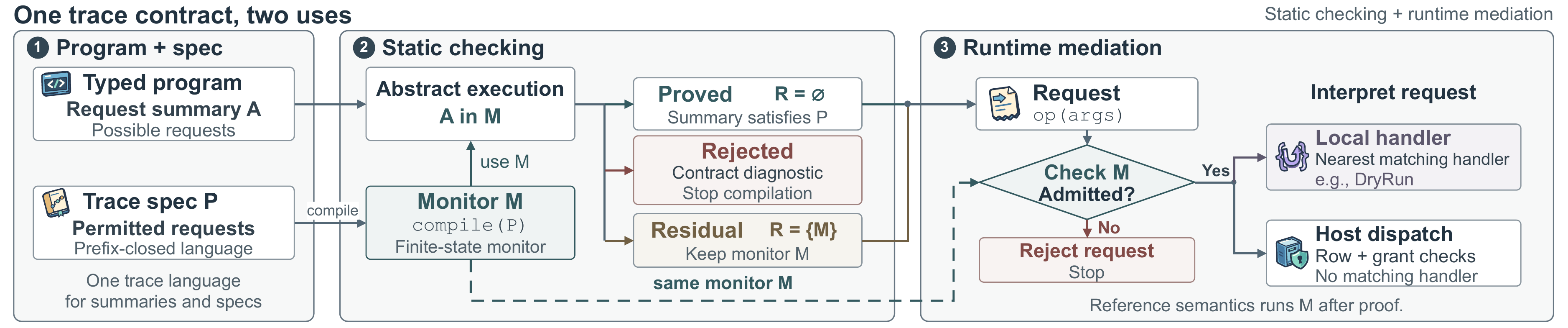}
\caption{Static and dynamic uses of an effect request trace contract.
(1) The effect request summary \(A\) describes possible effect request sequences and
the spec \(P\) the permitted ones. (2) Abstract execution of \(A\) in
\(M=\compile(P)\) proves conformance (\(R=\emptyset\)), rejects the
program, or leaves a runtime check (\(R=\{M\}\)).
(3) Even after proof, the reference semantics uses the same monitor
(dashed link) to check each effect request before handler selection. Host dispatch
also requires effect-row coverage and a grant.}
\Description{A wide three-panel diagram. Program and contract inputs produce
an effect request summary A and a compiled monitor M. Abstract execution branches to
Proved, Rejected, or Residual. Proved and Residual lead to runtime effect requests;
Rejected has no execution path. The same monitor, connected by a dashed
dependency, checks each effect request before interpretation. A denied effect request stops.
An admitted effect request enters the nearest matching local handler, such as DryRun,
or host dispatch with row and authorization checks.}
\label{fig:checking-pipeline}
\end{figure}

\subsection{Typing and Effect Request Summaries}
\label{sec:types}
\label{sec:effects}

\begin{wrapfigure}[14]{R}{0.50\textwidth}
\begin{minipage}{\linewidth}
\footnotesize
\[
\begin{alignedat}{2}
\langof 0&=\emptyset
 &\quad\langof{A;B}&=\langof A\,\langof B\\[1mm]
\langof{\emptytrace}&=\{\lambda_{\!w}\}
 &\quad\langof{A\sqcup B}&=\langof A\cup\langof B\\[1mm]
\langof\effop&=\{\effop\}
 &\quad\langof{A\cap B}&=\langof A\cap\langof B\\[1mm]
\langof{A^\star}&=(\langof A)^\star
 &\quad\langof{\msf{pref}(A)}&=\pref{\langof A}\\[1mm]
 & &\quad\langof{A\shuffleop B}
 &=\bigcup_{\substack{u\in\langof A\\v\in\langof B}}\msf{shuffle}(u,v).
\end{alignedat}
\]
\caption{Denotational semantics of the trace expressions in \figref{fig:core-syntax}.
The same interpretation maps request summaries \(A\) and contracts \(P\)
to languages over \(\Eff\). $\lambda_{\!w}$ is empty word.}
\Description{Denotational semantics of the shared effect request trace expressions for effect request trace. \lambda_{\!w} is empty word.}
\label{fig:trace-abstractions}
\end{minipage}
\end{wrapfigure}

The declarative typing judgment
\(\typeja{\Gamma}{e}{\tau}{\epsilon}{A}\) (\figref{fig:typing-rules})
assigns a computation \(e\) a result type
\(\tau\), an escaping-effect bound \(\epsilon\subseteq\Eff\), and an
ordered effect request summary \(A\) under \(\Gamma\). For syntactic
values, \(\valj\Gamma v\tau\) abbreviates
\(\typeja\Gamma v\tau\emptyset\emptytrace\), since evaluating a value
produces no effects or requests. Instead, a function
type \(\funty{\tau_1}{\tau_2}{\epsilon}{A}\) or handler type
\(\handty{\HandledOps}{\epsilon_h}{\tau_r}{A_h}\) records latent summaries that
are exposed upon invocation or handling.

Program summaries \(A\) and effect request trace contracts \(P\) share the trace syntax
of \figref{fig:core-syntax}, interpreted denotationally in
\figref{fig:trace-abstractions}. Here
\(\langof A\subseteq\Eff^\star\) overapproximates complete request words
of normally terminating executions, while its prefix closure
\(\pref{\langof A}\) covers finite prefixes before return. Sequential
composition uses complete words because a let initializer must return
before its body runs, whereas prefix closure captures partial executions
that may diverge or abort. To avoid colliding with the patch-application
operation \(w\), we write \(u\in\Eff^\star\) for general request words.

Typing rules compose escaping rows and request summaries in parallel.
An operation request \(\kw{perform}\ \effop(\overline v)\) adds
\(\effop\) to both (\textsc{T-Perform}). In \textsc{T-Let}, escaping
rows combine by union \(\epsilon_1\cup\epsilon_2\), whereas request
summaries compose sequentially as \(A_1;A_2\). Conditional branches join
summaries as \(A_1\sqcup A_2\) (\textsc{T-If}), while application
(\textsc{T-App}) and named calls (\textsc{T-Call}) expose the callee's
latent or declared bounds \((\epsilon_c,A_c)\). Recursive callables check
their bodies against supplied post-fixed-point summaries with nonempty
request languages, allowing conservative star summaries such as
\(\Eff^\star\) to cover finite request prefixes of nonterminating paths.
Rule \textsc{T-Weaken} widens either bound conservatively. On \mtt{apply}
(\figref{fig:overview-example}b), typing yields
\(A=g;(\emptytrace\sqcup w)\).

Specifications also constrain polymorphic instantiation. In \textsc{T-Var},
the instantiation operator \(\msf{Inst}_\DeclSig\)
computes admissible monomorphic instances over closed, well-formed types \(T\):
it yields \(\{\tau\}\) for a monotype, collects
\(\msf{Inst}_\DeclSig(\sigma[T/X])\) across all \(T\) for
\(\forall X.\sigma\), and restricts \(\forall X\sim S.\sigma\) to
those \(T\) satisfying \(\vdash T\sim S[T/X]\). Rule \textsc{T-TypeSpec}
discharges this obligation by looking up a validated witness
\(\DeclSig(T\sim S[T/X])=w_{\msf{ty}}\); if no witness exists, \(T\)
cannot instantiate the scheme and elaboration rejects an explicit generic
call at its use site (\figref{fig:elaboration}).

\emph{Example.} Given \etascode{spec Reviewable;},
\etascode{impl Patch ~ Reviewable;}, and
\etascode{flow id<T ~ Reviewable>(x: T) -> T ![] { return x; }},
elaboration records
\(\DeclSig(\msf{Patch}\sim\msf{Reviewable})=w_{\msf{ty}}\).
Calling \etascode{id<Patch>(p)} with \etascode{p: Patch} discharges
\(\vdash\msf{Patch}\sim\msf{Reviewable}\) via this witness, so
\textsc{T-Call} types the specialized Core call at \(\msf{Patch}\) with
summaries \((\emptyset,\emptytrace)\), whereas
\etascode{id<string>("diff")} fails because \(\DeclSig\) has no witness
for \(\TyString\sim\msf{Reviewable}\).

\begin{figure*}[t]
\footnotesize
\newcommand{\tyrule}[2]{\namedrule{#1}{#2}}
\let\inferrule\rulefraction
\noindent
\rulecell{.325\linewidth}{
\tyrule{T-Var}{\inferrule{x:\sigma\in\Gamma\quad\tau\in\msf{Inst}_\DeclSig(\sigma)}
 {\valj{\Gamma}{x}{\tau}}}
}\hfill
\rulecell{.325\linewidth}{
\tyrule{T-Prim}{\inferrule{\DeclSig(d)=\tau}
 {\valj{\Gamma}{d}{\tau}}}
}\hfill
\rulecell{.325\linewidth}{
\tyrule{T-Abs}{\inferrule{\typeja{\Gamma,x:\tau_1}{e}{\tau_2}{\epsilon}{A}}
 {\valj{\Gamma}{\lambda x.e}{(\funty{\tau_1}{\tau_2}{\epsilon}{A})}}}
}
\rulerowgap
\rulecell{.515\linewidth}{
\tyrule{T-App}{\inferrule{
 \valj{\Gamma}{v_f}{(\funty{\tau_1}{\tau_2}{\epsilon}{A})}\quad
 \valj{\Gamma}{v}{\tau_1}}
 {\typeja{\Gamma}{v_f\,v}{\tau_2}{\epsilon}{A}}}
}\hfill
\rulecell{.47\linewidth}{
\tyrule{T-Call}{\inferrule{
 \DeclSig(c)=\callty{\overline{\tau_i}}{\tau_o}{\epsilon_c}{A_c}\quad
 \forall i.\ \valj{\Gamma}{v_i}{\tau_i}}
 {\typeja{\Gamma}{\kw{call}\ c(\overline{v_i})}{\tau_o}
   {\epsilon_c}{A_c}}}
}
\rulerowgap
\rulecell{.515\linewidth}{
\tyrule{T-Let}{\inferrule{
 \typeja{\Gamma}{e_1}{\tau_1}{\epsilon_1}{A_1}\quad
 \typeja{\Gamma,x:\tau_1}{e_2}{\tau_2}{\epsilon_2}{A_2}}
 {\typeja{\Gamma}{\kw{let}\ x=e_1\ \kw{in}\ e_2}{\tau_2}
   {\epsilon_1\cup\epsilon_2}{A_1;A_2}}}
}\hfill
\rulecell{.47\linewidth}{
\tyrule{T-If}{\inferrule{
 \valj{\Gamma}{v}{\TyBool}\\
 \typeja{\Gamma}{e_1}{\tau}{\epsilon_1}{A_1}\quad
 \typeja{\Gamma}{e_2}{\tau}{\epsilon_2}{A_2}}
 {\typeja{\Gamma}{\kw{if}\ v\ \kw{then}\ e_1\ \kw{else}\ e_2}{\tau}
   {\epsilon_1\cup\epsilon_2}{A_1\sqcup A_2}}}
}
\rulerowgap
\rulecell{.325\linewidth}{
\tyrule{T-Weaken}{\inferrule{
 \typeja{\Gamma}{e}{\tau}{\epsilon}{A}\\
 \epsilon\subseteq\epsilon'\quad
 \langof A\subseteq\langof{A'}}
 {\typeja{\Gamma}{e}{\tau}{\epsilon'}{A'}}}
}\hfill
\rulecell{.325\linewidth}{
\tyrule{T-Perform}{\inferrule{
 \DeclSig(\effop)=\overline{\tau_i}\Rightarrow\tau_o\\
 \forall i.\ \valj{\Gamma}{v_i}{\tau_i}}
 {\typeja{\Gamma}{\kw{perform}\ \effop(\overline{v_i})}{\tau_o}
   {\{\effop\}}{\effop}}}
}\hfill
\rulecell{.325\linewidth}{
\tyrule{T-Spec}{\inferrule{
 \typeja{\Gamma}{e}{\tau}{\epsilon}{A}\\
 \vdash(\compty{\tau}{\epsilon}{A})\sim S\leadsto R}
 {\Gamma\vdash e\sim S:\tau\mathbin{!}\epsilon
   \mathbin{\triangleright}A\dashv R}}
}
\rulerowgap
\rulecell{.24\linewidth}{
\tyrule{T-TypeSpec}{\inferrule{
 \DeclSig(\tau\sim Y\,\overline{\tau_i})=w_{\msf{ty}}}
 {\vdash\tau\sim Y\,\overline{\tau_i}}}
}\hfill
\rulecell{.34\linewidth}{
\tyrule{T-TraceSpec}{\inferrule{
 \langof P=\pref{\langof P}\ne\emptyset\\
 M=\compile(P)\quad M\vdash A\leadsto R}
 {\vdash(\compty{\tau}{\epsilon}{A})\sim P\leadsto R}}
}\hfill
\rulecell{.40\linewidth}{
\tyrule{T-CallableSpec}{\inferrule{\epsilon\sqsubseteq\epsilon_s}
 {\vdash(\funty{\overline{\tau_i}}{\tau_o}{\epsilon}{A})
   \sim(\overline{\tau_i}\Rightarrow\tau_o\mathbin{!}\epsilon_s)}}
}
\caption{Declarative typing and spec checking.
\(\epsilon\) bounds escaping effects, while \(A\) summarizes effect requests,
including locally handled ones. Value judgments omit
\((\emptyset,\emptytrace)\). The last row checks named type, trace, and
callable constraints. An omitted callable effect bound means
\(\epsilon_s=\Eff\). \textsc{T-Spec} records residual obligations \(R\),
with \(R=\emptyset\) for a static proof and \(R=\{M\}\) for runtime
enforcement by monitor \(M\). All typing judgments range over well-formed types, type schemes, and environments.}
\Description{Typing of variables, data, functions, calls, effect requests,
composition, conditionals, and type, callable, and effect request trace constraints.}
\label{fig:typing-rules}
\end{figure*}

\subsection{Handlers}
\label{sec:static-handlers}

A handler type \(\handty{\HandledOps}{\epsilon_h}{\tau_r}{A_h}\) records
four interface components: handled operations \(\HandledOps\subseteq\Eff\),
the union of arm escaping-effect bounds \(\epsilon_h\), shared answer
type \(\tau_r\), and joined arm effect request summary \(A_h\). In
\figref{fig:handler-typing}, \textsc{T-Handler} synthesizes this type
from the handler's arms, while \textsc{T-HandleWith} applies a handler
value to an expression \(e\) of matching answer type \(\tau_r\),
transforming \(e\)'s escaping row \(\epsilon\) into
\((\epsilon\setminus\HandledOps)\cup\epsilon_h\): handled operations are
discharged from \(e\)'s row, while arm effects propagate outward.

For effect request tracking, \textsc{T-HandleWith} builds the conservative
summary \(\msf{pref}(A\shuffleop A_h^\star)\). Here, shuffle
(\(\shuffleop\)) interleaves body and arm effect requests while preserving each
operand's internal order, Kleene star (\(\star\)) covers any finite number
of arm invocations, and prefix closure (\(\msf{pref}\)) accounts for early
exits via \etascode{finish}. Handling may therefore expand the effect request
summary with arm behavior, while ensuring that every operation requested
by \(e\) remains visible to effect request trace contracts.

\begin{figure}[t]
\centering
\footnotesize
\thinmuskip=2mu\relax
\medmuskip=3mu\relax
\thickmuskip=3mu\relax
\setlength{\jot}{1pt}
\newcommand{\htrule}[4][l]{%
  \makebox[8em][#1]{\normalfont\textsc{#2}\enspace}%
  \(\displaystyle\rulefraction{#3}{#4}\)}
\begin{tabular}{@{}l@{}}
\htrule{T-Handler}{
  \forall i.\ \DeclSig(\effop_i)=\overline{\tau_i}\Rightarrow\tau_i^r\quad
 \forall i.\ \armj{\Gamma,\overline{x_i:\tau_i}}{k_i}{u_i}
     {\tau_i^r}{\tau_r}{\epsilon_i}{A_i}
}{
  \valj{\Gamma}{\kw{handler}\{\effop_i(\overline{x_i};k_i)\Rightarrow u_i\}_{i\in I}}
  {\handty{\HandledOps}{\bigcup_i\epsilon_i}{\tau_r}{\bigsqcup_i A_i}}
} \\[2pt]
\htrule{T-HandleWith}{
  \typeja{\Gamma}{e}{\tau}{\epsilon}{A}\quad
 \valj{\Gamma}{v_h}{\handty{\HandledOps}{\epsilon_h}{\tau}{A_h}}
}{
  \typeja{\Gamma}{\kw{handle}\ e\ \kw{with}\ v_h}{\tau}
  {(\epsilon\setminus \HandledOps)\cup\epsilon_h}
  {\msf{pref}(A\shuffleop A_h^\star)}
} \\[2pt]
\htrule{T-Continue}{
  \valj{\Gamma}{v}{\tau_a}
}{
  \armjshort{\Gamma}{k}{k\,v}{\tau_a}{\tau_r}
}\qquad
\htrule[r]{T-Return}{
  \valj{\Gamma}{v}{\tau_r}
}{
  \armjshort{\Gamma}{k}{v}{\tau_a}{\tau_r}
} \\[2pt]
\htrule{T-ArmLet}{
  \typeja{\Gamma}{e}{\tau_x}{\epsilon_1}{A_1}\quad
 \armj{\Gamma,x:\tau_x}{k}{u}{\tau_a}{\tau_r}{\epsilon_2}{A_2}
}{
  \armj{\Gamma}{k}{\kw{let}\ x=e\ \kw{in}\ u}{\tau_a}{\tau_r}
   {\epsilon_1\cup\epsilon_2}{A_1;A_2}
} \\[2pt]
\htrule{T-ArmIf}{
  \valj{\Gamma}{v}{\TyBool}\quad
 \armj{\Gamma}{k}{u_1}{\tau_a}{\tau_r}{\epsilon_1}{A_1}\quad
 \armj{\Gamma}{k}{u_2}{\tau_a}{\tau_r}{\epsilon_2}{A_2}
}{
  \armj{\Gamma}{k}{\kw{if}\ v\ \kw{then}\ u_1\ \kw{else}\ u_2}
   {\tau_a}{\tau_r}{\epsilon_1\cup\epsilon_2}
   {A_1\sqcup A_2}
}
\end{tabular}
\par
\caption{Typing of first-class handlers and terminal arms.
In \textsc{T-Handler}, \(\HandledOps=\{\effop_i\mid i\in I\}\), with pairwise
distinct \(\effop_i\). Each arm-bound \(k\) is fresh and absent from
\(\Gamma\), so only \textsc{T-Continue} can use it, in tail position,
with an operation result of type \(\tau_a\). Direct returns have type \(\tau_r\).
The summary \(\msf{pref}(A\shuffleop A_h^\star)\) includes handled
effect requests while admitting interleaved arm executions and early termination.}
\Description{Five rows of typing rules cover handler values, handler application,
continuation and direct return side by side, sequencing, and conditionals.}
\label{fig:handler-typing}
\end{figure}

To keep resumptions single-shot, the arm judgment
\(\armj\Gamma k u{\tau_a}{\tau_r}\epsilon A\) (\figref{fig:handler-typing},
abbreviated \(\armjshort\Gamma k u{\tau_a}{\tau_r}\) when
\(\epsilon=\emptyset\) and \(A=\emptytrace\)) distinguishes the operation
result type \(\tau_a\) accepted by \(k\) from the block's answer type
\(\tau_r\). An \etascode{Approval.request} arm in a string-returning block
may thus resume via \etascode{resume true;} (\textsc{T-Continue}) or exit
directly via \etascode{finish "denied";} (\textsc{T-Return}).
\textsc{T-ArmLet} excludes the fresh \(k\) from its initializer and
passes it only to the tail, while \textsc{T-ArmIf} checks its guard
without \(k\) and both branches under the same operation result and
answer types. These rules confine \(k\) to at most one tail call per
path. Both terminal
exits take values and contribute \((\emptyset,\emptytrace)\) locally,
since \textsc{T-HandleWith} already accounts for the resumed body.

Reusable handlers such as \mtt{DryRun} can generalize their answer type
at an immutable binding when they never exit directly. For a handler
value \(h\) typed by \textsc{T-Handler} as \(\tau_h\), local exits follow
\textsc{T-ArmLet} tails and \textsc{T-ArmIf} branches, excluding nested
handler arms. When every local exit tail-calls its resumption, we
generalize a fresh answer variable \(X\notin\msf{ftv}(\Gamma)\) with
\(\tau_h=\handty{\HandledOps}{\epsilon_h}{X}{A_h}\):
\[
\msf{gen}_{h}(\tau_h)=
\begin{cases}
\forall X.\tau_h & \text{if every local exit of \(h\) calls its resumption},\\
\tau_h & \text{otherwise}.
\end{cases}
\]
Binding \(\kw{let}\ x=h\ \kw{in}\ e\) assigns \(x\) the scheme
\(\msf{gen}_h(\tau_h)\), allowing instantiation at multiple answer types
in the first case. Having derived escaping rows and effect request summaries
across handlers, \secref{sec:spec-conformance} next checks them against
specifications.

\subsection{Spec Checking}
\label{sec:spec-conformance}
\label{sec:trace-specs}

The conformance judgment \(\vdash\tau\sim S\) (\figref{fig:typing-rules})
checks type and callable specifications. For named type constraints,
\textsc{T-TypeSpec} validates \(\vdash\tau\sim Y\,\overline{\tau_i}\)
against a witness \(w_{\msf{ty}}\) in \(\DeclSig\), gating constrained
polymorphism in \textsc{T-Var} (\secref{sec:types}). For callables,
\textsc{T-CallableSpec} matches parameter and return types while bounding
latent escaping effects by \(\epsilon\sqsubseteq\epsilon_s\)
(\(\epsilon_s=\Eff\) if omitted). Thus \etascode{Pure<string,string>}
elaborates to \((\TyString\Rightarrow\TyString\mathbin{!}\emptyset)\),
constraining invocation rather than closure construction.

Effect request trace specs instead constrain requested operations via
\(\vdash(\compty{\tau}{\epsilon}{A})\sim P\leadsto R\)
(\textsc{T-TraceSpec}), requiring \(\langof P\) to be nonempty and
prefix-closed and, for static discharge, verifying that every execution
prefix covered by \(A\) fits the contract:
\[
\pref{\langof A}\subseteq\langof P.
\]
In our review workflow, this condition tests \mtt{apply}'s summary
\(A=g;(\emptytrace\sqcup w)\) (\secref{sec:types}) against
\mtt{ApprovedApply}'s elaborated contract
\(P=\emptytrace\sqcup g;(g\sqcup w)^\star\) (\secref{sec:elaboration}).
Because every nonempty prefix of \(A\) begins with approval \(g\) and
\(P\) accepts the empty word,
\(A\) satisfies \(P\), whereas reordering the flow to \(w;g\) violates it.
Under \mtt{DryRun}, the write request \(w\) remains in
\(\msf{pref}(A\shuffleop A_h^\star)\) (\secref{sec:static-handlers})
and still undergoes the same check.

Finally, \textsc{T-Spec} attaches trace-spec checking to typing
derivations at computation boundaries and records the residual obligation
\(R\). Checking a callable body under its parameter context ensures
contracts govern invocation rather than value construction. 

\subsection{Monitor Construction and Abstract Checking}
\label{sec:policy-automata}
\label{sec:abstract-interpretation}

To enforce a trace contract \(P\), the compiler builds a deterministic
finite automaton \(M=\compile(P)=(Q,q_0,\delta,\Bad)\) whose states track
which request continuations remain valid after a prefix. Its states are
the left quotients \(u^{-1}\langof P=\{v\mid uv\in\langof P\}\) for
\(u\in\Eff^\star\), together with \(\emptyset\), with
\(q_0=\langof P\), rejecting sink \(\emptyset\)
(\(\Bad=\{\emptyset\}\)), and transition
\(\delta(q,\effop)=\{v\mid\effop v\in q\}\) (lifted to words as
\(\delta^\star\)). Because \(\langof P\) is regular, nonempty, and
prefix-closed over finite \(\Eff\), \(Q\) is finite and every state
outside \(\Bad\) accepts.
Section~\ref{sec:policy-safety} proves \(u\in\langof P\) iff
\(\delta^\star(q_0,u)\notin\Bad\).

Because a summary \(A\) covers many candidate traces from branches and
handlers, static checking lifts \(M\) to \emph{sets} of monitor states
(marked by \(\sharp\)).
The concrete domain \(\ConcDomain=\powerset{\Eff^\star}\) and abstract
domain \(\AbsDomain M=\powerset Q\) relate via
\(\alpha_M(W)=\{\delta^\star(q_0,u)\mid u\in W\}\) and
\(\gamma_M(U)=\{u\in\Eff^\star\mid\delta^\star(q_0,u)\in U\}\),
forming a \emph{Galois connection}
\((\ConcDomain,\subseteq)\galois{\alpha_M}{\gamma_M}
 (\AbsDomain M,\subseteq)\)~\cite{cousot1977abstract}. The transfer
function \(\asem A:\AbsDomain M\to\AbsDomain M\) computes the states
reachable from \(U\subseteq Q\) along complete words in \(\langof A\):
\[
\asem\effop(U)=\{\delta(q,\effop)\mid q\in U\}
\qquad\text{and}\qquad
\asem A(U)=\{\delta^\star(q,u)\mid q\in U,\ u\in\langof A\}.
\]
Sequencing composes transfer functions, \(\sqcup\) takes unions, and
iteration computes the least fixed point \(\mu V.\,(U\cup\asem A(V))\).
Intersection synchronizes both operand automata on the same word (rather
than intersecting state sets), \(\shuffleop\) interleaves transitions,
and \(\msf{pref}\) marks as accepting the states that can reach an original
accepting state (\appref{app:proofs}), so product search between the summary
automaton and \(M\) over finite \(\Eff\) and \(Q\) computes
\(\asem A(\{q_0\})=\alpha_M(\langof A)\) exactly.

Letting \(U=\asem A(\{q_0\})\), the checker either rejects \(A\) or
derives \(M\vdash A\leadsto R\). If \(U\cap\Bad=\emptyset\), every
word in \(\langof A\) stays outside \(\Bad\) (establishing
\(\pref{\langof A}\subseteq\langof P\) via the Galois connection and
prefix closure of \(\langof P\)), proving \(R=\emptyset\). If
\(\emptyset\ne U\subseteq\Bad\), every word in \(\langof A\) violates
\(P\) and the summary is rejected. Otherwise, \(U\) mixes accepting and
rejecting states, requiring a residual monitor \(R=\{M\}\) at runtime
(\secref{sec:dynamic-semantics}).

\emph{Example.} For \mtt{ApprovedApply}, let \(q_a=\delta(q_0,g)\) be
the state after an approval request. Then \(\delta(q_a,w)=q_a\), while
\(\delta(q_0,w)=\emptyset\). Both \(g\) and \(gw\) reach \(q_a\), so
\(g;(\emptytrace\sqcup w)\) yields \(R=\emptyset\). The word \(wg\)
reaches the absorbing sink and is rejected. For the mixed summary
\(A=(g;w)\sqcup(w;g)\),
\[
\asem A(\{q_0\})=\alpha_M(\{gw,wg\})=\{q_a,\emptyset\}.
\]
Both accepting and rejecting states occur, so the checker admits the
boundary with \(R=\{M\}\).

\section{Dynamic Semantics}
\label{sec:dynamic-semantics}

\begin{wrapfigure}[10]{R}{0.50\textwidth}
\footnotesize
\begingroup
\centering
\(\displaystyle
\begin{array}{@{}l@{\quad}l@{\;}c@{\;}l@{}}
\syntag{Effect Events}&\eta&::=&\reqevt{\effop}{\overline v}
 \mid\handleevt{\effop}{\overline v}{h}\\
&&\mid&\commitevt{\effop}{\overline v}
 \mid\denyevt{\effop}{\overline v}{j}\\[1mm]
\syntag{Contexts}&E&::=&[\,]\mid\kw{let}\ x=E\ \kw{in}\ e\\
&&\mid&\kw{handle}\ E\ \kw{with}\ h\\[1mm]
\syntag{Runtime Extension}&e&::=&\cdots\mid\msf{err}.
\end{array}
\)\par
\endgroup
\caption{Runtime syntax extending \figref{fig:core-syntax}.
Event superscripts \(\mathrm r,\mathrm h,\mathrm c,\mathrm d\) denote
admitted requests, handler dispatches, commits, and denials.
Subscripts identify the handler \(h\) or failure stage
\(j\in\{\mathrm r,\mathrm c,\mathrm f\}\).}
\Description{Effect events, Core evaluation contexts, and the terminal
error used by the small-step semantics.}
\label{fig:dynamic-syntax}
\end{wrapfigure}

Static checking supplies a typed body in \(\DeclSig\), an escaping-effect
bound \(\epsilon_b\), and a residual obligation \(R\). At runtime, each
active monitor checks each request
before local handling or host commit (guarded by \(\effop\in\epsilon_b\)
and a host grant). Execution steps configurations \(\langle e,t\rangle\),
pairing expression \(e\) with an append-only event trace \(t\).

\subsection{Evaluation and Terminal Handlers}
\label{sec:dynamic-expression-evaluation}
\label{sec:dynamic-handlers}

Execution steps a reducible configuration \(\langle e,t\rangle\) by
decomposing \(e\) into an evaluation context \(E\) and a redex
(\figref{fig:dynamic-syntax}). Because elaboration sequences non-value
operands into \(\kw{let}\)-bindings, \(E\) descends only into \(\kw{let}\)
initializers and \(\kw{handle}\) bodies. In \figref{fig:dynamic-evaluation},
\textsc{E-Pure} lifts event-trace-preserving steps \(\longrightarrow_0\) into
\(E\) without changing \(t\): \textsc{E-Beta} and \textsc{E-Let} substitute
values, \textsc{E-If} selects \(e_b\) for Boolean \(b\), \textsc{E-Call}
simultaneously substitutes value arguments for the parameters of the checked
curried body \(\DeclSig(\msf{body}(c))\) (e.g., entering \mtt{apply(p)})
with \(\DeclSig\) fixed, and \textsc{E-HandleValue} returns \(v\) from
\(\kw{handle}\ v\ \kw{with}\ h\), where a closed, well-typed handler
value is a literal \(h\).

For an admitted effect request \(\kw{perform}\ \effop(\overline v)\) with a
matching arm in scope, \textsc{E-Perform-Handle} selects the nearest
handler \(h\) with \((\effop(\overline x;k)\Rightarrow u)\in h\).
Writing \(\dom(h)\) for the operations handled by \(h\), we lift \(\dom\)
along \(E\):
{\small
\[
\dom([\,])=\emptyset,\quad
\dom(\kw{let}\ x=E\ \kw{in}\ e)=\dom(E),\quad
\dom(\kw{handle}\ E\ \kw{with}\ h)=\dom(E)\cup\dom(h).
\]
}
Here \(E\) surrounds the handled block and \(E'\) lies between the effect request
and \(h\), with \(\effop\notin\dom(E')\) ensuring the nearest match.
Simultaneous capture-avoiding substitution yields
\(u[\overline v/\overline x,v_k/k]\), where
\(v_k=\lambda z.\,\kw{handle}\ E'[z]\ \kw{with}\ h\).

Because \(u\) runs outside \(h\), its terminal exit governs control:
in \mtt{DryRun}'s \etascode{resume;}, \textsc{E-Beta} applies \(v_k\) to
unit, plugging unit into \(E'\) and reinstalling \(h\) (deep handling),
whereas \etascode{finish v;} returns \(v\) to the outer \(E\) and discards
\(v_k\). Resumption supplies an operation result of type \(\tau_a\), while
direct return supplies the block's answer of type \(\tau_r\)
(\secref{sec:static-handlers}). Lemma~\ref{lem:continuation-substitution}
preserves this single-shot terminal structure under substitution. For a
callable \(c\) with entry row \(\epsilon_b\) and monitor
\(M=\DeclSig(\msf{spec}(c))=\compile(P)\), a call pushes \(M\) alongside
enclosing monitors, admitted effect requests advance all active monitors, and
return pops only the callee's monitor. Capture suspends the continuation's
local monitor scopes, resumption restores them, and \etascode{finish}
discards them. Section~\ref{sec:actions-traces} next defines effect request
admission, commits, and denials.

\begin{figure*}[t]
\footnotesize
\newcommand{\dyrule}[2]{\namedrule{#1}{#2}}
\let\inferrule\rulefraction
\noindent
\rulecell{.325\linewidth}{
\dyrule{E-Beta}{\inferrule{\,}
 {(\lambda x.e)v\longrightarrow_0 e[v/x]}}
}\hfill
\rulecell{.325\linewidth}{
\dyrule{E-Let}{\inferrule{\,}
 {\kw{let}\ x=v\ \kw{in}\ e\longrightarrow_0 e[v/x]}}
}\hfill
\rulecell{.325\linewidth}{
\dyrule{E-If}{\inferrule{b\in\{\kw{true},\kw{false}\}}
 {\kw{if}\ b\ \kw{then}\ e_t\ \kw{else}\ e_f\longrightarrow_0 e_b}}
}
\rulerowgap
\rulecell{.325\linewidth}{
\dyrule{E-Call}{\inferrule{\DeclSig(\msf{body}(c))=\lambda\overline{x}.\,e}
 {\kw{call}\ c(\overline v)\longrightarrow_0 e[\overline v/\overline x]}}
}\hfill
\rulecell{.325\linewidth}{
\dyrule{E-HandleValue}{\inferrule{\,}
 {\kw{handle}\ v\ \kw{with}\ v_h\longrightarrow_0 v}}
}\hfill
\rulecell{.325\linewidth}{
\dyrule{E-Pure}{\inferrule{e\longrightarrow_0 e'}
 {\langle E[e],t\rangle\step\langle E[e'],t\rangle}}
}
\rulerowgap
\rulecell{\linewidth}{
\dyrule{E-Perform-Handle}{
 \inferrule{
 \delta(q_t,\effop)\ne\emptyset\quad
 \effop\notin\dom(E')\quad
 (\effop(\overline x;k)\Rightarrow u)\in h}
 {\langle E[\kw{handle}\ E'[\kw{perform}\ \effop(\overline v)]\ \kw{with}\ h],t\rangle
  \step\langle E[u[\overline v/\overline x,v_k/k]],
    t\cdot\reqevt{\effop}{\overline v}\cdot\handleevt{\effop}{\overline v}{h}\rangle}}
}
\rulerowgap
\rulecell{.325\linewidth}{
\dyrule{E-Request-Deny}{\inferrule{
 \vphantom{\effop\notin\dom(E)}\\
 \vphantom{\delta(q_t,\effop)\ne\emptyset}\\
 \delta(q_t,\effop)=\emptyset
 \vphantom{\effop(\overline v)\Downarrow_{\msf{host}}v}}
 {\begin{aligned}
  &\langle E[\kw{perform}\ \effop(\overline v)],t\rangle\\
  &\quad\step\langle\msf{err},t\cdot\denyevt{\effop}{\overline v}{\mathrm r}\rangle
  \end{aligned}}}
}\hfill
\rulecell{.325\linewidth}{
\dyrule{E-Perform-Commit}{\inferrule{
 \effop\notin\dom(E)\\
 \delta(q_t,\effop)\ne\emptyset\quad\effop\in\epsilon_b\\
 \msf{grant}(\effop(\overline v),t)\quad
 \effop(\overline v)\Downarrow_{\msf{host}}v}
 {\begin{aligned}
  &\langle E[\kw{perform}\ \effop(\overline v)],t\rangle\\
  &\quad\step\langle E[v],t\cdot\reqevt{\effop}{\overline v}\cdot\commitevt{\effop}{\overline v}\rangle
  \end{aligned}}}
}\hfill
\rulecell{.325\linewidth}{
\dyrule{E-Commit-Deny}{\inferrule{
 \effop\notin\dom(E)\\
 \delta(q_t,\effop)\ne\emptyset\\
 \effop\notin\epsilon_b\ \lor\ \neg\msf{grant}(\effop(\overline v),t)
 \vphantom{\effop(\overline v)\Downarrow_{\msf{host}}v}}
 {\begin{aligned}
  &\langle E[\kw{perform}\ \effop(\overline v)],t\rangle\\
  &\quad\step\langle\msf{err},t\cdot\reqevt{\effop}{\overline v}
    \cdot\denyevt{\effop}{\overline v}{\mathrm c}\rangle
  \end{aligned}}}
}
\caption{Small-step execution over \(\langle e,t\rangle\).
Here \(q_t\) is the monitor state after admitted requests in its scope,
and \(\emptyset\) denotes rejection.
In \textsc{E-Perform-Handle},
\(v_k=\lambda z.\kw{handle}\ E'[z]\ \kw{with}\ h\)
captures the delimited continuation.
Requests are checked before local handling or host dispatch.
Event notation is defined in \figref{fig:dynamic-syntax},
and \figref{fig:dynamic-auxiliary} gives the host-failure case.}
\Description{Operational rules for pure evaluation, handlers, effect requests,
commits, and enforcement failures, with explicit continuation substitution.}
\label{fig:dynamic-evaluation}
\end{figure*}

\subsection{Effect Requests, Commits, and Denials}
\label{sec:actions-traces}
\label{sec:runtime-policy-enforcement}

Before handler lookup or host dispatch, the runtime checks effect request admission
at every active monitor (\figref{fig:dynamic-evaluation}). For one checked
boundary, let \(\msf{reqs}(t)\in\Eff^\star\) extract operation symbols in
order from admitted effect request events \(\reqevt{\effop}{\overline v}\) observed
in its scope (\figref{fig:dynamic-syntax}), so the monitor state is
\(q_t=\delta^\star(q_0,\msf{reqs}(t))\). An effect request
\(\kw{perform}\ \effop(\overline v)\) is admitted at this boundary when
\(\delta(q_t,\effop)\ne\emptyset\). Under \mtt{ApprovedApply}, if
\mtt{Repo.apply} precedes any approval request,
\(\delta(q_t,\effop)=\emptyset\) and \textsc{E-Request-Deny} halts at
\(\msf{err}\), appending \(\denyevt{\effop}{\overline v}{\mathrm r}\) to
\(t\) without extending \(\msf{reqs}(t)\).

Following the approved path in \mtt{apply}, the event trace suffix depends on
whether a handler matches. Under \mtt{DryRun}, \textsc{E-Perform-Handle}
appends \(\reqevt{\effop}{\overline v}\cdot\handleevt{\effop}{\overline v}{h}\)
before the arm executes \etascode{resume;}. A matching arm that exits
immediately via \etascode{finish} gives the same suffix. When
\(\effop\notin\dom(E)\), the row and grant checks pass, and the host
succeeds (\(\effop(\overline v)\Downarrow_{\msf{host}}v\)),
\textsc{E-Perform-Commit} appends
\(\reqevt{\effop}{\overline v}\cdot\commitevt{\effop}{\overline v}\).
Both paths record \mtt{Repo.apply} in \(\msf{reqs}(t)\), whereas only a
successful host commit adds it to \(\msf{commits}(t)\), the set of
committed operation symbols.

Host execution adds an independent authorization layer: \mtt{ApprovedApply}
checks for a prior approval request, the conditional checks the Boolean
reply, and the trusted host's total predicate
\(\msf{grant}(\effop(\overline v),t)\) checks the concrete arguments and
execution history for target-path permissions and patch-specific consent.
Unhandled dispatch requires \(\effop\in\epsilon_b\) and
\(\msf{grant}(\effop(\overline v),t)\). Failing either triggers
\textsc{E-Commit-Deny} (\(\denyevt{\effop}{\overline v}{\mathrm c}\));
if both hold but \(\effop(\overline v)\Downarrow_{\msf{host}}\msf{fail}\),
\textsc{E-Host-Fail} records \(\denyevt{\effop}{\overline v}{\mathrm f}\)
(\figref{fig:dynamic-auxiliary}). Both rules record the admitted effect request
before the denial event and halt at \(\msf{err}\) with no commit.
Progress assumes every permitted, well-typed host request can return a
typed result or fail. Likewise, model inference validates outputs against
their schemas. Tool proposals return as data to the checked harness,
which validates names and arguments before issuing effect requests. Under complete
mediation, any resulting host dispatch requires the same row and grant checks.

Static checking and runtime enforcement apply the same monitor
\(M=\compile(P)\) to summary \(A\) and concrete prefix \(\msf{reqs}(t)\),
respectively. When \(R=\emptyset\), \(\pref{\langof A}\subseteq\langof P\)
holds, so a compiler preserving \(A\) and its scope may elide the effect request trace
check that the reference semantics runs uniformly. When \(R=\{M\}\),
runtime monitoring remains mandatory. Under either outcome, host commits
still require \(\effop\in\epsilon_b\) and \(\msf{grant}(\effop(\overline v),t)\).

\subsection{Recursive Self-Improvement with Checked Harnesses}
\label{sec:recursive-self-improvement}

Harness-level \emph{recursive self-improvement} (RSI) uses execution
feedback to revise an agent's harness~\cite{zhang2026selfharness,xia2026rrsi,wei2026selfhpo}.
\AgentLang{} isolates editable prompt, tool-selection, and post-processing
logic from fixed supervisor contracts: input/output types, escaping-effect
bound \(\epsilon_b\), effect request trace contract \(P\), and permitted
host bindings. Because candidate evaluation executes the harness and local
annotations cannot weaken supervisor contracts, every revision and its
dependencies must pass checking
(\secref{sec:elaboration}--\secref{sec:static-semantics}) before running
with \(\DeclSig\) fixed.

Checking reports repair-guiding diagnostics for rejected revisions, while
concrete host grants remain dynamic. Execution yields task scores and event traces
\(t\) to guide subsequent code or dependency edits, which trigger
rechecking. Contract checking thus decides which candidates may run, while
benchmark evaluation measures task performance (\secref{sec:evaluation-rsi}).
Section~\ref{sec:soundness} next establishes formal soundness.

\section{Soundness}
\label{sec:soundness}

Under a well-formed signature \(\DeclSig\), well-formed contracts \(P\),
and complete host mediation, we establish type safety
(\secref{sec:basic-type-safety}), effect and effect request bounds
(\secref{sec:effect-soundness}), trace-spec safety
(\secref{sec:policy-safety}), and handler transparency
(\secref{sec:handler-transparency}), with supporting lemmas and proofs in \appref{app:proofs}.

\subsection{Preservation and Progress}
\label{sec:basic-type-safety}
\label{sec:soundness-aux}

To separate ordinary value typing from changes in effect and effect request
summaries during evaluation, \appref{app:proofs} erases latent \(\epsilon\)
and \(A\) annotations while keeping runtime monitoring intact and typing
\(\msf{err}\) at any expected type. Every source derivation maps to this
erased relation. At \textsc{E-Perform-Handle}, the captured continuation
\(v_k\) has ordinary type \(|\tau_a|\to|\tau_r|\), so substituting it for
\(k\) in the checked arm yields type \(|\tau_r|\)
(\lemref{lem:continuation-substitution}). Erased preservation and progress
thus ensure ordinary type safety, while
\secref{sec:effect-soundness}--\secref{sec:handler-transparency} establish
effect and effect request trace safety.

\begin{theorem}[Preservation]\label{thm:preservation}
Under a well-formed signature and typed host relation, if a closed runtime
expression \(e\) has erased type \(\tau\) and
\(\langle e,t\rangle\step\langle e',t'\rangle\), then \(e'\) has
erased type \(\tau\), including the terminal error \(\msf{err}\) at that type.
\end{theorem}
\begin{proof}[Proof sketch]
Induct on the step derivation. Value substitution and context replacement
cover pure reductions. In \textsc{E-Perform-Handle}, the captured lambda
has type \(\tau_a\to\tau_r\), and
Lemma~\ref{lem:continuation-substitution} types the substituted arm at
\(\tau_r\). Typed host results and \textsc{R-Err} cover the remaining
cases (Appendix~\ref{app:proofs}).
\end{proof}

\begin{theorem}[Progress]\label{thm:progress}
A closed, well-typed reachable configuration is a value, the terminal
error \(\msf{err}\), or has a next step, assuming each permitted host operation can return
a typed result or fail.
\end{theorem}
\begin{proof}[Proof sketch]
Canonical forms and context decomposition expose a pure redex or effect request.
Monitor rejection, the nearest matching handler, or the host
authorization and result/failure cases supply a step.
Lemma~\ref{lem:continuation-substitution} excludes free resumptions
from evaluation positions.
\end{proof}

\subsection{Effect and Effect Request Bounds}
\label{sec:effect-soundness}

Type preservation alone does not constrain which operations execution may
request or commit. Theorem~\ref{thm:effect-safety} bridges typing and
execution by confining host commits to the derived escaping row
\(\epsilon\) (which may be tighter than the boundary \(\epsilon_b\)) and
confining admitted effect request words to \(\pref{\langof A}\) (or
\(\langof A\) upon return of a value).

\begin{theorem}[Effect and effect request soundness]\label{thm:effect-safety}
Let a closed expression have judgment
\(\typeja{\emptyset}{e}{\tau}{\epsilon}{A}\), with
\(\epsilon\subseteq\epsilon_b\).
For every finite run from \(\langle e,\lambda_{\!w}\rangle\) to
\(\langle e',t\rangle\):
\[
 \msf{reqs}(t)\in\pref{\langof A}
 \qquad\text{and}\qquad
 \msf{commits}(t)\subseteq\epsilon.
\]
If the run returns a value, its effect request word belongs to \(\langof A\).
\end{theorem}
\begin{proof}[Proof sketch]
Induct on finite executions using checked recursive interfaces
(Appendix~\ref{app:summary-proof}). Value substitution preserves latent
bounds, and lets concatenate complete words with unfinished prefixes.
Handling preserves body effect request order while inserting arm words and
possibly truncating a suffix, giving
\(\msf{pref}(A\shuffleop A_h^\star)\).
The escaping row \((\epsilon\setminus \HandledOps)\cup\epsilon_h\) removes only
intercepted body operations and includes arm effects, so every commit
is covered. Denials make no commit and end a bounded effect request prefix.
\end{proof}

\subsection{Trace-Spec Safety}
\label{sec:policy-safety}

Theorem~\ref{thm:effect-safety} reduces effect request trace conformance to checking
the summary \(A\). Lemma~\ref{lem:trace-elaboration}
(\appref{app:trace-elaboration}) ensures the elaborated spec \(P\) matches
the source clauses, while Lemmas~\ref{lem:monitor} and~\ref{lem:abstract}
prove that the monitor \(M=\compile(P)\) recognizes \(\langof P\) and
that a successful abstract check establishes acceptance of all prefixes
covered by \(A\). Theorem~\ref{thm:policy-safety} combines these static
bounds with runtime monitoring when \(R=\{M\}\).

\begin{lemma}[Monitor correctness]\label{lem:monitor}
For every well-formed contract \(P\) and effect request word \(w\),
\(w\in\langof P\) iff
\(\delta^\star_{\compile(P)}(q_0,w)\notin\Bad\).
\end{lemma}
\begin{proof}[Proof sketch]
Induction on \(w\) gives
\(\delta^\star(q_0,w)=w^{-1}\langof P\).
A nonempty quotient admits an extension of \(w\), so prefix closure
implies \(w\in\langof P\). Conversely, an accepted word has the empty
continuation.
\end{proof}

\begin{lemma}[Abstract checking]\label{lem:abstract}
If \(\asem A(\{q_0\})\cap\Bad=\emptyset\), then every
word in \(\pref{\langof A}\) is accepted by \(M\).
\end{lemma}
\begin{proof}[Proof sketch]
Automaton reachability gives
\(\asem A(\{q_0\})=\alpha_M(\langof A)\).
With the premise, the Galois connection and Lemma~\ref{lem:monitor} yield
\(\langof A\subseteq\gamma_M(Q\setminus\Bad)=\langof P\),
whose prefix closure gives the result. Intersection uses a synchronized
operand product (Appendix~\ref{app:proofs}).
\end{proof}

\begin{theorem}[Trace-spec safety with residual checks]\label{thm:policy-safety}
Every finite execution under \(M=\compile(P)\) for well-formed \(P\),
starting with the empty event trace, satisfies \(\msf{reqs}(t)\in\langof P\). This holds for
statically discharged and residual contracts. If a contract is statically
discharged and its checked effect request bound is preserved, removing its
effect request monitor does not change successful executions or introduce a
violating effect request.
\end{theorem}
\begin{proof}[Proof sketch]
Induct on steps, observing that every appended effect request passes the monitor, and effect request
denial leaves the admitted word unchanged. Lemma~\ref{lem:monitor}
gives spec satisfaction. For monitor removal,
Theorem~\ref{thm:effect-safety} and Lemma~\ref{lem:abstract} ensure
that the check cannot fail. Authorization and effect-boundary checks
remain unchanged.
\end{proof}

\subsection{Handler Transparency}
\label{sec:handler-transparency}

Theorem~\ref{thm:handler-transparency} guarantees that even when a local handler replaces a host commit with a local result, the triggering effect request remains in the event trace and handling cannot bypass monitor, boundary, or grant checks.

\begin{theorem}[Persistent handled effect requests]\label{thm:handler-transparency}
Every event \(\handleevt{\effop}{\overline v}{h}\) in a reachable event trace has an
immediately preceding \(\reqevt{\effop}{\overline v}\). Both remain in every extension
of that event trace. Installing or invoking a handler cannot create a commit
without the ordinary monitor, boundary, and grant premises.
\end{theorem}
\begin{proof}[Proof sketch]
\textsc{E-Perform-Handle} appends the adjacent effect-request/dispatch pair,
and no rule deletes events. Arm effect requests undergo the same checks.
Only \textsc{E-Perform-Commit} appends commits, subject to the unchanged
monitor, row, and grant premises.
\end{proof}

Together, Theorems~\ref{thm:preservation}--\ref{thm:handler-transparency}
ensure that for every well-typed model or host result, admitted effect request
words belong to \(\langof P\) and host commits respect \(\epsilon\) and
\(\msf{grant}\). Section~\ref{sec:implementation} next describes how our
Rust compiler, optimizer, and interpreter realize these checks and apply
agent-specific optimizations in practice.

\section{Implementation}
\label{sec:implementation}

We implement \AgentLang{} toolchain in Rust including a frontend, an optimizer, and an interpreter. The frontend parses source programs, and lowers computations to the high level intermediate representation (HIR), where it checks types, escaping effects, requested operations, and specifications. optimizer then perform  agent specific optimization, and the interpreter executes the resulting HIR.

\subsection{Checking and Compilation}
\label{sec:implementation-pipeline}

The frontend lowers agent, tools, and flows into callable HIR signatures with parameter, result, and effect interfaces, collecting them before checking bodies for recursion, discharging substituted type-spec bounds at generic call sites, and reading checked imported metadata. Call-graph fixed-point analysis tracks escaping effects separately from requested operations, removing handled operations from the escaping row and adding arm effects while keeping their requests visible. Trace checking combines these request facts with ordered control flow to emit source-mapped diagnostics for forbidden requests or missing prerequisites. Successful checking emits a checked HIR snapshot for the optimizer (\secref{sec:implementation-optimization}) and interpreter.

\begin{table}[t]
\centering
\captionsetup{font=scriptsize}
\begin{minipage}[t]{0.39\linewidth}
\captionsetup{type=table}
\caption{Survey of 17 fixed-revision repositories, numbered R1--R17
for Table~\ref{tab:optimization-opportunities}. Project names link to
surveyed commits.}
\label{tab:benchmark-sources}
\end{minipage}\hfill
\begin{minipage}[t]{0.59\linewidth}
\captionsetup{type=table}
\caption{Optimization opportunities across the repository IDs in
Table~\ref{tab:benchmark-sources}. 
Counts partition 2,801 candidates. Forms denote selected HIR rewrites
or model-context passes P1--P4.}
\label{tab:optimization-opportunities}
\end{minipage}\par
\begin{minipage}[b]{0.39\linewidth}
\centering\scriptsize
\setlength{\tabcolsep}{1.5pt}
\begin{tabular*}{\linewidth}{@{\extracolsep{\fill}}rlrrr@{}}
\toprule
\textbf{ID} & \textbf{Repository} & \textbf{Files} & \textbf{Lines} & \textbf{Cand.} \\
\midrule
\hypertarget{survey-repo-1}{\textsf{R1}} & \href{https://github.com/agent0ai/agent-zero/tree/ff2a5b9c9c4e490ba38dd7c313e779510bb2071a}{Agent Zero}~\cite{agentzero2026repository} & 614 & 70,523 & 255 \\
\hypertarget{survey-repo-2}{\textsf{R2}} & \href{https://github.com/Aider-AI/aider/tree/5dc9490bb35f9729ef2c95d00a19ccd30c26339c}{aider}~\cite{aider2026repository} & 114 & 21,582 & 21 \\
\hypertarget{survey-repo-3}{\textsf{R3}} & \href{https://github.com/browser-use/browser-use/tree/2cc71890e65056b67fdb29d98bb407749c15bb08}{Browser Use}~\cite{browseruse2026repository} & 155 & 56,691 & 87 \\
\hypertarget{survey-repo-4}{\textsf{R4}} & \href{https://github.com/OpenBMB/ChatDev/tree/a2ea4a263cebe9a7abd01ee3b5b0bf95698ed328}{ChatDev}~\cite{chatdev2026repository} & 115 & 21,898 & 51 \\
\hypertarget{survey-repo-5}{\textsf{R5}} & \href{https://github.com/cline/cline/tree/51f63e4fd63aeeb74d708b9c3d1f8b800aff9b3e}{Cline}~\cite{cline2026repository} & 409 & 125,827 & 121 \\
\hypertarget{survey-repo-6}{\textsf{R6}} & \href{https://github.com/openai/codex/tree/02ca1e0d990e496366bda7c30e2981cf296783ce}{Codex}~\cite{codex2026repository} & 1,234 & 425,151 & 174 \\
\hypertarget{survey-repo-7}{\textsf{R7}} & \href{https://github.com/google-gemini/gemini-cli/tree/718e97c2448e0384cb582409cf4e479437bf50fe}{Gemini CLI}~\cite{geminicli2026repository} & 298 & 87,634 & 497 \\
\hypertarget{survey-repo-8}{\textsf{R8}} & \href{https://github.com/aaif-goose/goose/tree/c85a58d7c71033d4a0ffdb3419a9f948f1254cf9}{Goose}~\cite{goose2026repository} & 488 & 547,048 & 457 \\
\hypertarget{survey-repo-9}{\textsf{R9}} & \href{https://github.com/FoundationAgents/MetaGPT/tree/11cdf466d042aece04fc6cfd13b28e1a70341b1f}{MetaGPT}~\cite{metagpt2026repository} & 505 & 55,253 & 122 \\
\hypertarget{survey-repo-10}{\textsf{R10}} & \href{https://github.com/SWE-agent/mini-swe-agent/tree/04d809ceab9df28f9adaed044884180159172930}{mini-swe-agent}~\cite{miniswe2026repository} & 44 & 4,683 & 19 \\
\hypertarget{survey-repo-11}{\textsf{R11}} & \href{https://github.com/openclaw/openclaw/tree/76985e4c31c4c94c91e7ed79411c8189180b7ea8}{OpenClaw}~\cite{openclaw2026repository} & 6,343 & 1,223,858 & 473 \\
\hypertarget{survey-repo-12}{\textsf{R12}} & \href{https://github.com/anomalyco/opencode/tree/382929112664ecfa1776f36e089c1f165fce1868}{OpenCode}~\cite{opencode2026repository} & 624 & 109,069 & 25 \\
\hypertarget{survey-repo-13}{\textsf{R13}} & \href{https://github.com/OpenHands/OpenHands/tree/76ed10527eaf8ef8e106368d5fc3df94d6242497}{OpenHands}~\cite{openhands2026repository} & 28 & 7,666 & 2 \\
\hypertarget{survey-repo-14}{\textsf{R14}} & \href{https://github.com/OpenHands/software-agent-sdk/tree/07e8c26ea404b9597bfdf1a10e8ad8e97b811839}{OpenHands SDK}~\cite{softwareagentsdk2026repository} & 470 & 127,350 & 139 \\
\hypertarget{survey-repo-15}{\textsf{R15}} & \href{https://github.com/FoundationAgents/OpenManus/tree/3309bf4e416fb1c74b008f3e86494439a31bad53}{OpenManus}~\cite{openmanus2025} & 65 & 19,181 & 18 \\
\hypertarget{survey-repo-16}{\textsf{R16}} & \href{https://github.com/earendil-works/pi/tree/e7a0df27c09d11bf4161d3605246cbf57a129cb2}{Pi}~\cite{pi2026repository} & 298 & 66,558 & 321 \\
\hypertarget{survey-repo-17}{\textsf{R17}} & \href{https://github.com/SWE-agent/SWE-agent/tree/3ea751c087f32b16e039a2233dd6eefecef325d5}{SWE-agent}~\cite{sweagent2026repository} & 68 & 11,430 & 19 \\
\midrule
\multicolumn{2}{l}{Total} & 11,872 & 2,981,402 & 2,801 \\
\bottomrule
\end{tabular*}

\end{minipage}\hfill
\begin{minipage}[b]{0.59\linewidth}
\centering\scriptsize
\setlength{\tabcolsep}{1.5pt}
\renewcommand{\arraystretch}{1}
\setlength{\extrarowheight}{-0.636pt}
\begin{tabular*}{\linewidth}{@{\extracolsep{\fill}}
  >{\raggedright\arraybackslash}p{0.30\linewidth}r
  >{\raggedright\arraybackslash}p{0.25\linewidth}
  >{\raggedright\arraybackslash}p{0.27\linewidth}@{}}
\toprule
\textbf{Opportunity} & \textbf{Cand.} & \textbf{Repos.} & \textbf{Form} \\
\midrule
Request preparation & 406 & \mbox{\hyperlink{survey-repo-1}{\textsf{R1}}--\hyperlink{survey-repo-12}{\textsf{R12}}}, \mbox{\hyperlink{survey-repo-14}{\textsf{R14}}--\hyperlink{survey-repo-17}{\textsf{R17}}} & Hoisting, P2 \\
\addlinespace[2pt]
Versioned reuse & 348 & \mbox{\hyperlink{survey-repo-1}{\textsf{R1}}}, \mbox{\hyperlink{survey-repo-3}{\textsf{R3}}--\hyperlink{survey-repo-9}{\textsf{R9}}}, \mbox{\hyperlink{survey-repo-11}{\textsf{R11}}--\hyperlink{survey-repo-14}{\textsf{R14}}}, \mbox{\hyperlink{survey-repo-16}{\textsf{R16}}--\hyperlink{survey-repo-17}{\textsf{R17}}} & Pure-expression reuse \\
\addlinespace[2pt]
Context and observation demand & 415 & \mbox{\hyperlink{survey-repo-1}{\textsf{R1}}--\hyperlink{survey-repo-17}{\textsf{R17}}} & Demand, P1, P3, P4 \\
\addlinespace[2pt]
Collection and control specialization & 683 & \mbox{\hyperlink{survey-repo-1}{\textsf{R1}}--\hyperlink{survey-repo-12}{\textsf{R12}}}, \mbox{\hyperlink{survey-repo-14}{\textsf{R14}}--\hyperlink{survey-repo-17}{\textsf{R17}}} & Fusion, indexing \\
\addlinespace[2pt]
Allocation and snapshot sharing & 756 & \mbox{\hyperlink{survey-repo-1}{\textsf{R1}}--\hyperlink{survey-repo-12}{\textsf{R12}}}, \mbox{\hyperlink{survey-repo-14}{\textsf{R14}}--\hyperlink{survey-repo-17}{\textsf{R17}}} & Copy removal \\
\addlinespace[2pt]
Serialization fusion & 73 & \mbox{\hyperlink{survey-repo-1}{\textsf{R1}}}, \mbox{\hyperlink{survey-repo-3}{\textsf{R3}}--\hyperlink{survey-repo-4}{\textsf{R4}}}, \mbox{\hyperlink{survey-repo-6}{\textsf{R6}}--\hyperlink{survey-repo-12}{\textsf{R12}}}, \mbox{\hyperlink{survey-repo-14}{\textsf{R14}}--\hyperlink{survey-repo-16}{\textsf{R16}}} & JSON fusion \\
\addlinespace[2pt]
Scheduling and batching & 61 & \mbox{\hyperlink{survey-repo-3}{\textsf{R3}}--\hyperlink{survey-repo-6}{\textsf{R6}}}, \mbox{\hyperlink{survey-repo-9}{\textsf{R9}}--\hyperlink{survey-repo-11}{\textsf{R11}}}, \mbox{\hyperlink{survey-repo-14}{\textsf{R14}}--\hyperlink{survey-repo-15}{\textsf{R15}}} & Fusion, batching \\
\addlinespace[2pt]
Runtime-resource reuse & 56 & \mbox{\hyperlink{survey-repo-2}{\textsf{R2}}--\hyperlink{survey-repo-11}{\textsf{R11}}}, \mbox{\hyperlink{survey-repo-14}{\textsf{R14}}}, \mbox{\hyperlink{survey-repo-16}{\textsf{R16}}} & Factory hoisting \\
\addlinespace[2pt]
\midrule
Test-harness reuse & 3 & \mbox{\hyperlink{survey-repo-9}{\textsf{R9}}} & Survey only \\
\addlinespace[2pt]
\midrule
\textbf{Total} & \textbf{2,801} & \mbox{\hyperlink{survey-repo-1}{\textsf{R1}}--\hyperlink{survey-repo-17}{\textsf{R17}}} & \\
\bottomrule
\end{tabular*}

\end{minipage}
\Description{Two separately captioned tables showing surveyed repositories
on the left and optimization opportunity classes on the right.
Repository IDs link the opportunity counts to their source projects.}
\end{table}

\subsection{Agent-Specific Optimization}
\label{sec:implementation-optimization}
Real-world agent workflows contain many optimizable patterns, such as rebuilding unchanged retry payloads, repeatedly scanning message histories, and forwarding unneeded tool schemas. Leveraging the agent abstraction, checked types and effect bounds from \secref{sec:implementation-pipeline}, our optimizer turns these patterns into compiler passes.

\subsubsection{Optimization Opportunities}

To identify optimization opportunities in agent workflows, we surveyed 17 open-source agent repositories (11,872 files; 2,981,402 lines), summarized in Table~\ref{tab:benchmark-sources}. The survey identifies 2,801 conditional optimization candidates---each recording source coordinates and rewrite preconditions---and classifies their underlying mechanisms into 30 families. 

Table~\ref{tab:optimization-opportunities} groups 24 of the 25 candidate-bearing mechanism families into eight opportunity classes, with the remaining test-harness reuse family reported separately. For example, invariant retry payloads motivate hoisting, while repeated message-sequence scans motivate scan fusion. These recurring opportunities inform the design of HIR rewrites and context-selection configurations.

\subsubsection{Typed-HIR Rewrite}
The HIR optimizer uses typed effect to perform optimizations, including local simplification at O1, scan fusion at O2, and resource-lifetime rewrites at O3. Optimization levels select cumulative pass groups, with O0 leaving HIR unoptimized. These rewrites are conditioned on strict invariant checks (e.g., input purity, disjoint accumulators) to ensure execution correctness.

\subsubsection{Context Selection}
The compiler exposes four independently selectable model-context passes. P0 serves as the full-context baseline with all four disabled. P1 compacts model history, P2 restricts tool catalogs to workflow-relevant definitions, P3 projects observations to required fields, and P4 slices source context to task-relevant dependencies. P1--P4 each enable only the named pass, while \mtt{All} combines all four. The interpreter applies the selected reductions when preparing native model requests, independently of the local-execution optimization levels O0--O3.

\subsection{Execution and Rechecking}
\label{sec:implementation-runtime}
\label{sec:implementation-diagnostics}

The interpreter executes checked HIR by preparing slot layouts and dispatch metadata, dispatching matching requests to local handlers and permitted unhandled operations to host services, validating model outputs and proposed tool names and arguments, and updating host grant state separately from approval requests. For reflective harness evolution, checked HIR and analysis facts bind to each project revision, so each edit to prompt construction, tool selection, result processing, or dependencies must recompile against fixed entrypoint types, effect bounds, and specs (with diagnostics guiding rejection or repair) before the interpreter executes the new snapshot. Section~\ref{sec:evaluation} evaluates these checks, optimizations, and revisions.

\section{Evaluation}
\label{sec:evaluation}

We evaluate \AgentLang{} with four research questions:
\begin{itemize}
  \item \textbf{RQ1: Safety contracts.} Are effect bounds and trace specs checked as intended, including effect requests interpreted by local handlers?
  \item \textbf{RQ2: Agent-specific optimization, performance.} How much execution work do typed-HIR rewrites eliminate?
  \item \textbf{RQ3: Agent-specific optimization, tokens.} How does context selection affect model tokens and task correctness?
  \item \textbf{RQ4: Recursive self-improvement.} How do successive self-revisions of a financial and a patent agent affect question-answering accuracies and the size of the revised code?
\end{itemize}

\paragraph{Experimental Environment.}
We run the experiments on an x86-64 server with two Intel Xeon Gold 6444Y processors, each with 16 cores, running Linux 5.15.0-191-generic and glibc 2.35. 
The compiler and interpreter are built with Rust 1.99.1, using Rust \mtt{opt-level=3} for every \AgentLang{} HIR optimization level. 
The token experiment uses \mtt{MiniMax-M3.1-Flash} through a remote provider.

\subsection{Safety Contracts in Nomo (RQ1)}
\label{sec:evaluation-cases}

We used \AgentLang{} to implement a coding agent, Nomo. Nomo currently contains 120 source files and 7,671 lines of code. %
We evaluate its read controller, which chooses file reads or source searches from model replies, using 24 application files (1,417 lines). Keeping the controller's decision logic unchanged, we supply model replies and tool results through local handlers. The production loop requires 3 host-service families (model, filesystem, and console), compared with 1 under local handling (console), while all 3 handled operation kinds remain in the static request summary.

Across 10 deterministic scenarios, 4 complete normally and 6 refuse at the expected boundary, with all 10 matching the expected controller states and request histories. Local handlers successfully interpret all 19 requests (12 model, 5 read, and 2 search) with 0 mishandled. \AgentLang{} checks locally handled requests against the same trace specs as requests dispatched to the host. In standalone tests (\secref{sec:evaluation-compiler}), the compiler rejects all three mutations involving missing approval, reversed request order, or a forbidden write intercepted by a local handler.

\paragraph{Answer to RQ1.}
Both incompatible implementations are rejected before execution, one for an incorrect return type and one for an escaping effect outside the declared bound. In 6 standalone compiler fixtures (\secref{sec:evaluation-compiler}), both valid programs pass and all 4 mutations are rejected, covering operation order, forbidden handled requests, and escaping-effect bounds. Together, the application results and contract fixtures show that the compiler successfully perform all checks.

\subsection{Agent-Specific Optimization (RQ2 and RQ3)}
\label{sec:evaluation-optimization}

We evaluate optimizations across two dimensions: \emph{typed-HIR rewrites} that reduce local execution work, and \emph{context-selection configurations} that prune model inputs.

\subsubsection{Workloads}
\label{sec:evaluation-workloads}

We identified recurring computational patterns in the open-source repositories and adapted them into synthetic \AgentLang{} programs to isolate the effects of HIR rewrites. For local execution benchmarks (\secref{sec:evaluation-runtime}), we synthesize 26 workflows (231--311 effective lines each, 6,965 total), each combining six of 13 pattern families with distinct normalized shapes and step/allocation profiles, collectively triggering 22 rewrite kinds. For token optimization benchmarks (\secref{sec:evaluation-tokens}), we synthesize 12 two-stage workflows (e.g., repository review, incident triage, release gating, and checkpoint recovery) with 80--144-line task modules and drivers exposing 8--19 typed tools.

\subsubsection{Execution Performance}
\label{sec:evaluation-runtime}

To measure performance, we execute each workflow 21 times per optimization level (after three warm-up iterations), using a local scripted host to isolate HIR execution from model and network latency. Timing excludes compilation, planning, and setup overhead. The harness validates every result and effect summary against the unoptimized O0 baseline to confirm that the program functionality remains unchanged. %

\begin{figure}[t]
\centering
\includegraphics[width=\linewidth]{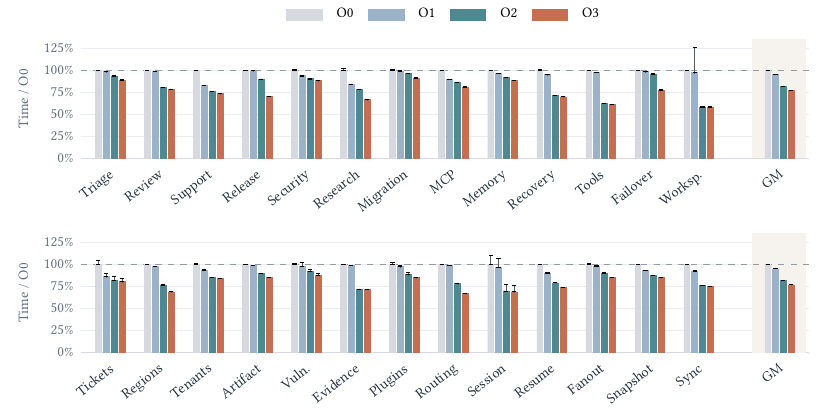}
\caption{HIR execution time across 26 workflows at O0--O3.
Bars show each workflow's median divided by its O0 median over 21 runs,
and whiskers show the normalized run interquartile range.
The shaded GM group in both panels summarizes all 26 workflows by
geometric mean. Hosts are local and scripted. Lower is better.}
\label{fig:optimization-results}
\Description{Grouped runtime bars for all 26 workflows, split across
two panels. Gray, blue, teal, and orange indicate O0 through O3.
Whiskers show variation across runs. All workflows have lower median
execution time at O3 than at O0.}
\end{figure}

Figure~\ref{fig:optimization-results} presents workflow execution times normalized to the O0 baseline. There are cumulative performance gains of 5.2\%, 18.4\%, and 23.2\% at O1, O2, and O3, respectively. Similar to the O3 performance improvement, there was a 22.8\% reduction in interpreter steps and a 24.6\% decrease in memory allocations. All 26 workflows gain speedups, with median-time improvements ranging from 8.9\% to 42.2\%. Re-validation confirms that these rewrites preserve all original return values and effect summaries.

\paragraph{Answer to RQ2.}
Typed-HIR rewrites reduce execution time across all 26 workflows. The 23.2\% time reduction at O3 is positively correlated with reduced computation.

\subsubsection{Token Optimization}
\label{sec:evaluation-tokens}

\begin{figure}[t]
\centering
\includegraphics[width=\linewidth]{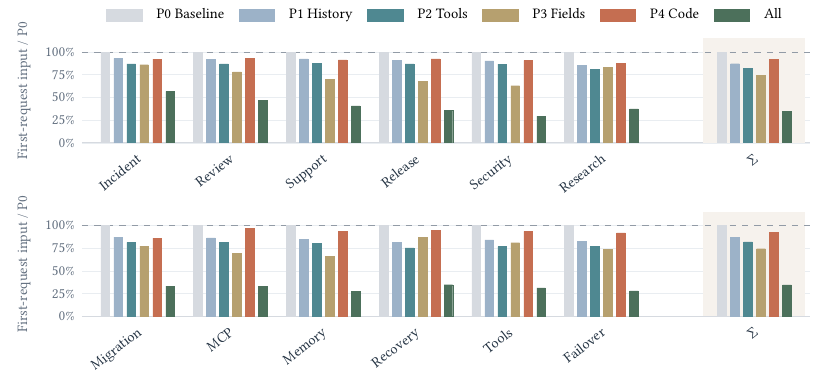}
\caption{Provider-reported input tokens for the first model request in
12 \AgentLang{} workflows. Bars show single-execution ratios to P0, with
individual passes P1--P4 and their combination. The shaded $\Sigma$ group
uses summed first-request tokens over all 12 workflows, including
failures. Lower is better.}
\label{fig:token-optimization-results}
\Description{Grouped first-request input-token ratios for 12 workflows
in two rows sharing one legend. Six colors encode the baseline, four individual passes,
and all passes. Every optimized first
request uses fewer input tokens than its baseline.}
\end{figure}

We execute each workflow under six configurations: (1) the baseline configuration P0, (2) four individual optimization passes P1--P4 that remove redundant context history, tool calls, data fields, and code, respectively, and (3) a combined optimization pass \mtt{All}. These experiments produce 72 executions and 147 model calls. Table~\ref{tab:token-results} and \figref{fig:token-optimization-results} present the results.

\begin{wraptable}{R}{0.38\textwidth}
\centering\scriptsize
\captionsetup{font=scriptsize}
\caption{Token optimization passes on 12 workflows. Token reductions relative
to P0 use full-execution sums over workflows correct under both
configurations.}
\label{tab:token-results}
\setlength{\tabcolsep}{1.2pt}
\begin{tabular*}{\linewidth}{@{\extracolsep{\fill}}lrrr@{}}
\toprule
Configuration & Input $\downarrow$ & Total $\downarrow$ & Uncached $\downarrow$ \\
\midrule
P0 Baseline & 0.0\% & 0.0\% & 0.0\% \\
P1 History & 12.3\% & 10.9\% & 8.8\% \\
P2 Tools & 16.7\% & 16.1\% & 20.9\% \\
P3 Fields & 31.4\% & 28.9\% & 26.4\% \\
P4 Code & 7.1\% & 5.9\% & 3.3\% \\
All Passes & 67.8\% & 65.2\% & 70.8\% \\
\bottomrule
\end{tabular*}

\end{wraptable}

Combining all four passes reduces input tokens by 67.8\%, total tokens by 65.2\%, and uncached input by 70.8\%. Among individual passes, P3 gives the largest input reduction at 31.4\%. The combined configuration satisfies the oracle on all 12 workflows when compared with P0. Across all 72 executions, five fail the oracle, three because of an extra model round and two because of incorrect report fields. Including failures, input and total tokens decrease by 68.9\% and 66.3\%, respectively, across all 12 workflows under the combined configuration.

\paragraph{Answer to RQ3.}
The four passes reduce first-request input in all 60 optimized executions. Their combination satisfies the task oracle on all 12 workflows and reduces full-execution input by 67.8\%.

\subsection{Recursive Self-Improvement (RQ4)}
\label{sec:evaluation-rsi}

Harness-level RSI modifies an agent's core task-processing harness. We decompose it into four components: (1) \mtt{Prompt}, which defines domain-specific rules and output formats, (2) \mtt{Read}, which retrieves data from various sources like tables and documents, (3) \mtt{Solve}, which coordinates the high-level reasoning workflow, and (4) \mtt{Model}, which performs model inference. 
This decomposition allows the data extraction logic (\mtt{Read}) and the reasoning strategy (\mtt{Solve}) to be revised separately. 

We evaluate two \AgentLang{} agent harnesses, each over a baseline (H0) and five generations of revisions (H1--H5). The first is a Financial-QA agent, across 256 questions from TAT-QA\cite{zhu-etal-2021-tat}, categorized into span, multi-span, arithmetic, and count tasks. The second is an Patent agent evaluated on 30 questions from PANORAMA~\cite{panorama2025} and 20 close-ended questions from IPBench~\cite{ipbench2026}.

\usetikzlibrary{patterns}

\begin{figure}[t]
\centering
\includegraphics[width=\linewidth]{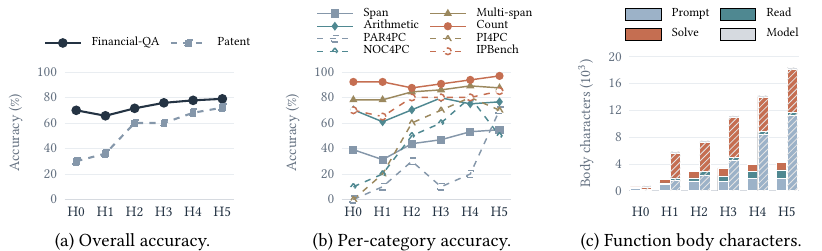}
\caption{
Self improvement of a Financial-QA harness (solid) and of a Patent harness (dashed). Each harness is evaluated on four task categories.
Overall accuracy (a) sums the four category counts. 
Category accuracy (b) uses 64 questions per category for the Financial-QA harness; 10 questions per category from PANORAMA (PAR4PC, PI4PC, NOC4PC) and 20 questions from IPBench for the Patent harness. 
Function-body sizes (c) count Unicode characters, including whitespace and line breaks, for Prompt, Read, Solve, and Model, with the left bar of each pair showing the Financial-QA harness and the right bar showing the Patent harness.}
\label{fig:rsi-results}
\Description{Three side-by-side panels show overall accuracy, accuracy for four question categories, and paired stacked bars of function-body character counts across H0 through H5. 
The overall and category panels draw two harnesses: the Financial-QA harness (solid, 256 questions) and the Patent harness (dashed, 50 questions). 
Financial-QA overall accuracy falls at H1 from 69.9 to 65.6 percent and rises to 78.9 percent at H5; Patent accuracy starts at 30.0 percent and rises to 72.0 percent at H5. 
The reported body size of the Financial-QA harness increases from 582 to 4234 characters, while the Patent total body size grows from 770 to 18449 characters, mostly prompt text.}
\end{figure}

Figure~\ref{fig:rsi-results} presents accuracy improvements in both harnesses. The Financial-QA agent improves its accuracy from 69.9\% in H0 to 78.9\% in H5 (+9.0\%).
All four categories improve from H0 to H5, with span-based questions showing the largest gain, from 39.1\% to 54.7\% (+15.6\%). The Patent agent improves from 30\% to 72\% (+42\%). All four task categories improve from H0 to H5. The agent initially fails almost entirely (0\% to 10\%) in PANORAMA tasks, and successive revisions significantly increase accuracies on these tasks.

The harness code size expanded over revisions. For Financial-QA, total non-blank lines grow from 8 to 58, and character counts increase from 582 to 4,234. For Patent, character counts grow from 770 to 18,449. In both harnesses, the \mtt{Model} body remains unchanged (29 characters for Financial-QA, 309 for Patent), while \mtt{Prompt}, \mtt{Read}, and \mtt{Solve} expanded. These results indicate that harness revisions occurred mainly in task processing logic and model prompts.

\paragraph{Answer to RQ4.}
RSI in \AgentLang{} improves overall Financial-QA accuracy by 9.0\% and Patent agent accuracy by 42\%. The revisions expand the function bodies \mtt{Prompt}, \mtt{Read}, and \mtt{Solve}.

\section{Related Work}
\label{sec:related-work}

Algebraic effects separate operation interfaces from their
interpretations \cite{plotkin2003algebraic,plotkin2013handling,
bauer2015programming}. Koka combines row-typed effects with practical
inference and compilation \cite{leijen2014koka,leijen2017typedirected};
Links and Frank study related direct-style and effect-polymorphic
programming models \cite{hillerstrom2016liberating,lindley2017dobedobedo}.
Handler names and capability-based accounts refine the connection
between operations and their handlers
\cite{xie2022firstclass,brachthaeuser2022effects}.
\AgentLang{} adopts first-class handlers with terminal single-shot
resumptions and a shared answer type. Its additional interface is an
effect request summary that survives handling and remains subject to effect request trace
conformance across real and simulated interpretations.

Security automata characterize enforcement by observing and stopping
executions \cite{schneider2000enforceable}. \AgentLang{} compiles trace
specs over typed effect requests into such automata.
Abstract interpretation provides the static counterpart
\cite{cousot1977abstract}, in which a powerset analysis over monitor states
establishes acceptance for an effect request summary, or requires runtime
monitoring. The language-specific point is which events are summarized
and checked. Effect requests precede both local interpretation and default
dispatch, while commits additionally require host authorization.
Monitoring only the final escaping row cannot recover handled effect requests.

LangGraph mixes deterministic and agentic graph nodes and supports
persistence and human intervention \cite{langgraphDocs}; AutoGen
provides conversational and event-driven agent programming
\cite{autogenDocs}; the OpenAI Agents SDK provides tools, guardrails,
tracing, and human-in-the-loop integration \cite{openaiAgentsSDK}.
These mechanisms can implement the workflow's protections.
\AgentLang{} makes operation bounds and effect request ordering constraints part
of component interfaces and checks them when computations are composed,
including across local handling.

Information-flow control constrains dependencies between data and
observations \cite{denning1976lattice}. Session types describe
communication protocols \cite{honda1998language}.
\AgentLang{} trace specs describe permitted effect requests and
their order, such as excluding writes or requiring approval before
application.

\section{Conclusion}
\label{sec:conclusion}

Viewing agent systems as programs, \AgentLang{} gives operation requests, local handler interpretation, and host commits distinct, compiler-checked semantics. Trace specs constrain requested operations and their order across both live services and local handlers, while checked intermediate facts guide optimization and harness rechecking. Our core calculus establishes preservation, progress, effect coverage, trace-spec safety, and persistent handled requests under typed host interfaces and complete mediation. Our Rust prototype and empirical evaluation confirm that these static and runtime checks enforce behavioral contracts while enabling measurable execution, token-context, and harness-revision gains.

\bibliographystyle{ACM-Reference-Format}
\bibliography{reference}

\appendix
\section{Supporting Definitions and Proofs}
\label{app:proofs}

We give the supporting arguments for \secref{sec:soundness} under a
fixed, well-formed \(\DeclSig\). Recursive bodies satisfy their declared
summaries, and host primitives produce no unreported effects or requests.
Model-selected tools execute through checked calls and performed operations.
Figure~\ref{fig:dynamic-auxiliary} supplies the host-failure transition.

\begin{figure}[!ht]
\footnotesize
\rulecell{\linewidth}{
\namedrule{E-Host-Fail}{\rulefraction{
 \effop\notin\dom(E)\quad
 \delta(q_t,\effop)\ne\emptyset\quad
 \effop\in\epsilon_b\quad\msf{grant}(\effop(\overline v),t)\quad
 \effop(\overline v)\Downarrow_{\msf{host}}\msf{fail}}
 {\langle E[\kw{perform}\ \effop(\overline v)],t\rangle
  \step\langle\msf{err},t\cdot\reqevt{\effop}{\overline v}
    \cdot\denyevt{\effop}{\overline v}{\mathrm f}\rangle}}
}
\caption{Host failure, supplementing \figref{fig:dynamic-evaluation}.
An unhandled, admitted, and authorized effect request may fail before committing;
the event trace records the admitted effect request and host failure before termination.}
\Description{The host-failure transition before an external commit.}
\label{fig:dynamic-auxiliary}
\end{figure}

\subsection{Type Erasure and Continuation Substitution}

Erasure \(|\tau|\) drops latent effect and request annotations while
preserving base and nominal types, function interfaces, and handler
operation sets and answer types (\figref{fig:runtime-typing}).
We write \(\Gamma\vdash e:T\) and
\(\Gamma;k\vdash u:T_a\Rightarrow T_r\) for erased expression and arm
typing, with \(T=|\tau|\) and \(\Gamma\) erased pointwise.

\begin{figure}[!ht]
\footnotesize
\[
 \left|\funty{\tau_1}{\tau_2}{\epsilon}{A}\right|
   =|\tau_1|\to|\tau_2|
 \qquad
 \left|\handty{\HandledOps}{\epsilon_h}{\tau_r}{A_h}\right|
   =\msf{H}(\HandledOps,|\tau_r|).
\]
\ruledivider
\[
\namedrule{R-Err}{\frac{\ }{\Gamma\vdash\msf{err}:T}}
\]
\caption{Type erasure (top) and terminal-error typing (bottom).
\(\msf{H}(\HandledOps,T_r)\) is the ordinary handler type with handled operations \(\HandledOps\)
and answer type \(T_r\). Captured continuations use ordinary function
typing. The terminal error \(\msf{err}\) has any expected type;
diagnostic details remain in the event trace.}
\Description{Erasure of latent type annotations and typing of terminal errors.}
\label{fig:runtime-typing}
\end{figure}

The runtime relation types \(\msf{err}\) at any expected type and
captured continuations as ordinary functions. Type schemes instantiate
after their constraints are discharged. Spec checking preserves the
underlying computation, with residual monitors enforced at its boundary.
Alongside \textsc{T-Let}, immutable handler bindings use
\(\msf{gen}_h\) from \secref{sec:static-handlers} as follows.
\[
\frac{
 \valj{\Gamma}{h}{\tau_h}
 \quad \sigma=\msf{gen}_h(\tau_h)
 \quad\typeja{\Gamma,x:\sigma}{e}{\tau}{\epsilon}{A}}
 {\typeja{\Gamma}{\kw{let}\ x=h\ \kw{in}\ e}{\tau}{\epsilon}{A}}.
\]

\begin{lemma}[Substitution and contexts]\label{lem:substitution}
If \(\Gamma,x:T\vdash e:U\) and \(\Gamma\vdash v:T\), then
\(\Gamma\vdash e[v/x]:U\). The corresponding statement holds for arm
typing, with the distinguished resumption fresh for the substituted value.
Value substitution preserves value operands. Replacing an evaluation
context's hole by an expression of the same erased type preserves the
enclosing derivation.
\end{lemma}
\begin{proof}
Simultaneous induction on expression and arm typing preserves parameter,
result, and answer types. Freshening term binders and generalized answer
variables avoids capture, while substitution maps values to values.
Context replacement follows by induction over let initializers and
handled bodies.
\end{proof}

\begin{lemma}[Continuation substitution]\label{lem:continuation-substitution}
If \(\Gamma;k\vdash u:T_a\Rightarrow T_r\), where
\(k\notin\dom(\Gamma)\), and \(\Gamma\vdash v_k:T_a\to T_r\), then
\(\Gamma\vdash u[v_k/k]:T_r\).
Substitution preserves value operands, and each local terminal path
returns directly or tail-calls \(v_k\) once.
\end{lemma}
\begin{proof}
Induct on arm typing. \textsc{T-Continue} becomes the ordinary application
\(v_k\,v:T_r\), and \textsc{T-Return} already supplies \(v:T_r\).
In \textsc{T-ArmLet}, weakening types \(v_k\) under the fresh let binder,
and the induction hypothesis types the tail. \textsc{T-ArmIf} preserves
the common branch type. Since ordinary premises exclude \(k\), free
occurrences appear only as terminal callees. Initializers, closures, and
nested handler arms remain free of the outer \(k\), so each selected
path invokes \(v_k\) at most once.
\end{proof}

For Theorem~\ref{thm:preservation}, inversion at
\textsc{E-Perform-Handle} gives operation result type \(T_a\) and shared
answer type \(T_r\). Replacing the request by fresh \(z:T_a\) types
\(\kw{handle}\ E'[z]\ \kw{with}\ h:T_r\), hence
\(v_k=\lambda z.\kw{handle}\ E'[z]\ \kw{with}\ h:T_a\to T_r\).
Value and continuation substitution type the selected arm at \(T_r\),
and context replacement preserves the outer type. The remaining cases
use value substitution, typed host results, or \textsc{R-Err}.
For an \(n\)-argument callable, the first \(n-1\) curried applications
are pure, and \textsc{E-Call} substitutes all arguments into the checked
innermost body.

For Theorem~\ref{thm:progress}, canonical forms and context decomposition
expose a pure redex or request. Closed functions are abstractions,
Boolean values select branches, and handler values supply their declared
arms. A request is denied by the monitor, intercepted by the nearest
matching handler, or dispatched through the row, grant, and host
result/failure checks. Continuation substitution supplies an ordinary
lambda at every resumption call.

\subsection{Effect and Request Bounds}
\label{app:summary-proof}

For Theorem~\ref{thm:effect-safety}, the induction tracks complete request
words in \(\langof A\), finite prefixes in \(\pref{\langof A}\), and
host commits in \(\epsilon\). An arm's local contribution ends at its
direct return or tail call, with subsequent requests assigned to the
resumed body. Ordinary value substitution preserves full type and summary
annotations by induction on typing, freshening generalized answer
variables at handler bindings.

Induct simultaneously over finite executions of expressions and local
arm computations, indexed by unfolding and application steps for
recursive and higher-order calls. The prefix argument requires nonempty
request languages for callable and latent interfaces. Values contribute
the empty word, and an admitted request contributes its operation symbol.
Calls expose checked bounds. A let contributes an initializer prefix or
a complete initializer followed by a body prefix, bounded by
\(A_1;A_2\) because the remaining summary is nonempty.
Conditionals use the selected branch, and weakening follows by language
inclusion.

For \textsc{T-HandleWith}, body requests preserve their order while
invoked arms insert local request words. Repeated arm invocations are
bounded by \(A_h^\star\), and direct return truncates the suspended body.
Thus every finite request word belongs to
\(\pref{\langof A\shuffleop(\langof{A_h})^\star}\).
Computing the handler operand occurs in an enclosing let and is accounted
for separately. Body operations in \(\HandledOps\) are intercepted, while
arms execute outside their own delimiter and may issue further escaping
requests, including operations in \(\HandledOps\). The escaping bound is
therefore \((\epsilon\setminus\HandledOps)\cup\epsilon_h\).

Checked recursive summaries bound each finite unfolding, and each body
entry reduces the induction index. Denial and host failure terminate a
bounded request prefix without adding a commit.

\subsection{Trace-Spec Elaboration}
\label{app:trace-elaboration}

An expanded spec denotes alternatives of conjunctive atomic clauses.
A word satisfies an alternative when every operation is allowed by a
positive or ordering clause, none is denied, and every \(W\) request in
\etascode{g >> w} has a strictly earlier \(G\) request. Here \(G,W\)
are resolved operation sets, and \etascode{w << g} has the same meaning.

\begin{lemma}[Trace-spec elaboration]\label{lem:trace-elaboration}
If an expanded trace spec \(S\) over the finite alphabet \(\Eff\)
translates to \(P\), then \(u\in\langof P\) iff \(u\) satisfies
an alternative of \(S\). The language is regular, nonempty, and
prefix-closed.
\end{lemma}
\begin{proof}
The ordering translation satisfies
\[
\begin{aligned}
&\langof{(\Eff\setminus W)^\star;
  (\emptytrace\sqcup(G\setminus W);\Eff^\star)}\\
&\qquad=\{\effop_1\cdots\effop_n\in\Eff^\star\mid
  \forall i\le n.\ \effop_i\in W\Rightarrow
  \exists j<i.\ \effop_j\in G\}.
\end{aligned}
\]
A word on the left either omits \(W\) or encounters a \(G\setminus W\)
before its first \(W\). Conversely, strict precedence supplies exactly
this factorization, including when \(G,W\) overlap or are empty.

Induct on the constructor count of the pending clauses. Starting from
\((\emptyset,\Eff^\star)\), \(U\) collects allowances and \(P\) enforces
denials and precedence. Positive clauses extend \(U\), denials intersect
\(P\) with words omitting their targets, and ordering clauses update both
using the identity above. Conjunction processes both operands, while
disjunction copies the pending suffix into each alternative before taking
their union. Reverse ordering uses the same identity.
Finalizing \(U^\star\cap P\) therefore enforces the allowances and
restrictions of a complete source alternative. Every atomic restriction is regular,
prefix-closed, and contains the empty word, and finite union and
intersection preserve these properties.
\end{proof}

\subsection{Finite Monitor Analysis}

Automata for choice, concatenation, and iteration use standard epsilon
constructions. Shuffle uses an asynchronous product, while intersection
synchronizes symbol edges and permits independent epsilon steps, accepting
when both operands accept. Prefix closure marks states that can reach an
original accepting state. These constructions preserve regularity.
Identifying deterministic states by their continuation languages yields
the quotients in \(\compile(P)\), with the empty language as an absorbing
rejecting sink.

The product of the summary automaton and \(M\) consumes matching operation
symbols, while epsilon steps leave \(M\) unchanged. Structural induction
therefore gives \(\asem A(\{q_0\})=\alpha_M(\langof A)\).
For intersection, synchronization enforces the same word in both operands.
Since the rejecting sink is absorbing, accepting every complete summary
word also accepts every prefix, establishing Lemma~\ref{lem:abstract}.

\section{Compiler Contract Mutations}
\label{sec:evaluation-compiler}

The compiler fixtures test whether effect request persistence is reflected in
source-level diagnostics. We invoke the \mtt{check} command on six
isolated fixtures using the installed \AgentLang{} 0.1.0 compiler on macOS/arm64.
Every fixture shares typed approval and repository actions.
The mutations change only order, handling, or a declared row.
Table~\ref{tab:contract-results} reports the observed outcomes.

\begin{table}[htbp]
\centering\small
\caption{Observed compiler outcomes for six isolated fixtures. Both valid
fragments are accepted; four mutations are rejected with the listed
diagnostic. These are compiler checks, not model-driven executions.
The artifact includes the complete inputs and outputs.}
\label{tab:contract-results}
\begin{tabular}{@{}lp{0.59\linewidth}@{}}
\toprule
\textbf{Fixture} & \textbf{Outcome} \\
\midrule
Approval before application & Accepted \\
Same flow with dry-run handler & Accepted \\
Approval omitted & \mtt{TraceSpecRequirementNotDominating} \\
Approval moved after application & \mtt{TraceSpecRequirementNotDominating} \\
Handled write under read-only spec & \mtt{TraceSpecAllowViolation} \\
Unhandled write with empty row & \mtt{EffectOutsideDeclaredRow} \\
\bottomrule
\end{tabular}
\end{table}

The compiler accepts both the approval-before-application fragment and
its dry-run variant. Missing or late approval triggers a dominance
diagnostic, a forbidden locally handled effect request triggers an allow
violation, and an unhandled write outside the declared row triggers an
effect-row diagnostic. All six outcomes match expectations.

The handled-forbidden case isolates the contribution of persistent
effect requests by showing that the compiler rejects a forbidden write even after a local
handler removes its escaping effect. The formal
\(\msf{pref}(A\shuffleop A_h^\star)\) bound includes an empty path and
can residualize the handled-forbidden case, whereas the compiler detects
its definite effect request and rejects it. The compiler therefore distinguishes
this case more precisely than the conservative formal bound. The fixtures
exercise effect request persistence with the prototype's analysis, while the
independent model exercises the formal monitor and summary constructions.

\section{Monitor and Summary Validation}
\label{sec:evaluation-model-checks}

An independent Python model lowers source clauses to the effect request trace
expressions of \figref{fig:core-syntax} and constructs deterministic
monitors. A positional oracle checks allow, deny, and strict precedence
directly on words over approval, write, and read. Enumerating all allow
and deny subsets with six precedence configurations gives 384 contracts.
All 419,712 comparisons against the 1,093 words of length at most six
agree with the oracle.

We additionally test 32 selected pairs of contracts under whole-word
disjunction and independent conjunction, giving 69,952 comparisons.
All agree with the corresponding union or intersection denotation.
A regression check distinguishes clause merging from language
intersection by showing that \(+g\mathbin{\&}+w\) admits \(gw\), although the languages
of the two closed contracts \(+g\) and \(+w\) do not both admit it.

Twelve directed summary cases produce the expected safe, rejected, or
residual outcome. Eleven cases check contract well-formedness, and six
accepted effect request trace expressions give 6,558 further word comparisons.
Four checks distinguish differently instantiated read operations.

\end{document}